%% file: fair_knapsack.tex
\documentclass[11pt]{article}
\usepackage[dvipsnames]{xcolor}
\usepackage[left=1in,right=1in,top=1in,bottom=1in]{geometry}
\usepackage[english]{babel}
\usepackage[utf8]{inputenc}
\usepackage{amsmath}
\usepackage{amsthm}
\usepackage{amsfonts}
\usepackage{graphicx}
\usepackage[colorinlistoftodos]{todonotes}
\usepackage{verbatim} 
\usepackage{enumitem}
\usepackage{fetamont}
\usepackage[T1]{fontenc}
\usepackage{times}
\usepackage{fullpage}
\usepackage{tikz,pgfplots}
\usetikzlibrary{patterns}
\usepackage{caption, setspace, subcaption}

\usepackage{xspace}

\usepackage{amssymb}
\usepackage[colorlinks = true, linkcolor = black, citecolor = black, final]{hyperref}
\usepackage{mathtools}
\usepackage{multicol}
\usepackage{ marvosym }
\usepackage{comment}
\usepackage{wasysym}
\usepackage{tikz}
\usetikzlibrary{patterns}
\usepackage{optidef}
\usepackage{physics}
\usepackage{cleveref}
\usepackage[english]{babel}
 
 \usepackage{float}
 
 \usepackage{adjustbox}
 \usepackage{multirow}
 
 \usepackage[ruled,vlined]{algorithm2e}

\newcommand{\argmax}{{\sf argmax}}
\newcommand{\argmin}{{\sf argmin}}
 
\usepackage{array}
\newcolumntype{P}[1]{>{\centering\arraybackslash}p{#1}}
\newcolumntype{M}[1]{>{\centering\arraybackslash}m{#1}}

\newtheorem{problem}{Problem}
 
\newtheorem{theorem}{Theorem}[section]

\newcommand{\NoOfCategories}{m}
\newcommand{\NoOfItmes}{N}

\def\DEBUG{true}
\ifdefined\DEBUG

\def\rem#1{{\marginpar{\raggedright\scriptsize #1}}}
\newcommand{\devr}[1]{\rem{\textcolor{red}{$\bullet$ #1}}}
 \newcommand{\arir}[1]{\rem{\textcolor{blue}{$\bullet$ #1}}}
\newcommand{\anar}[1]{\rem{\textcolor{green}{$\bullet$ #1}}}
\else

  \newcommand{\devr}[1]{}
  \newcommand{\anar}[1]{}
  \newcommand{\arir}[1]{}
\fi

\newcommand{\PTAS}{\textsc{PTAS}}
\newcommand{\FPTAS}{\textsc{FPTAS}}

\newcommand{\NP}{\textsc{NP}}
\newcommand{\PP}{\textsc{P}}

\newcommand{\eps}{\varepsilon}
\newcommand{\OPT}{{\sf OPT}}

\begin{document}
\title{Group Fairness for Knapsack Problems}

\author{Deval Patel \\IISc\\ \href{mailto:devalpatel@iisc.ac.in}{devalpatel@iisc.ac.in} \and Arindam Khan\\IISc\\ \href{mailto:arindamkhan@iisc.ac.in}{arindamkhan@iisc.ac.in}\and Anand Louis\\IISc\\ \href{mailto:anandl@iisc.ac.in}{anandl@iisc.ac.in}}
\date{}
\maketitle

\begin{abstract}
	
We study the knapsack problem with {\em group fairness} constraints. The input of the problem consists of a knapsack of bounded capacity and a set of items, each item belongs to a particular category and has an associated {\em weight} and {\em value}. The goal of this problem is to select a subset of items such that all categories are {\em fairly} represented, the total weight of the selected items does not exceed the capacity of the knapsack, and the total value is maximized. We study the fairness parameters such as the bounds on the total value of items from each category, the total weight of items from each category, and the total number of items from each category. We give approximation algorithms for these problems. These \textit{fairness} notions could also be extended to the min-knapsack problem.
The \textit{fair} knapsack problems encompass various important problems, such as participatory budgeting, fair budget allocation, advertising. 
\end{abstract}

\input{introduction}

\input{relatedwork}

\input{algorithm_max_knapsack}

\input{algorithm_min_knapsack}

\input{open_problems}

\bibliographystyle{alpha}
\bibliography{bibfile}

\end{document}

%% file: introduction.tex
\section{Introduction}
The knapsack (also known as {\em max knapsack}) problem is a classical packing problem. The input of the max-knapsack problem consists of a bounded capacity knapsack and a set of items, each having an associated {\em weight} and {\em value}. The objective is to select a subset of items such that the total weight of the selected items does not exceed the capacity of the knapsack and the total value is maximized. The {\em min knapsack}  problem is an extensively studied variant, where the input consists of a set of items, each having an associated {\em weight} and {\em value}, along with the lower bound on the packing value. The goal of this problem is to select a subset of items such that the total value of the selected items is at least the given bound, and the total weight is minimized. 
Max knapsack and min knapsack problems represent the prototypical packing and covering problems, respectively. \par

From the practical viewpoint, the knapsack problem models many prominent industrial problems such as budgeting, cargo packing, resource allocation, assortment planning, etc.~\cite{Keller}. The knapsack problem and its variants are  
special cases of many salient optimization problems, e.g., the generalized assignment problem, the packing and covering linear programs, etc.
They are also key subproblems in the solution of more complex problems, such as the \textit{cutting stock problem} \cite{Vanderbeck}. 

In this work, we consider the notion of \textit{group fairness} in knapsack problems. In this setting, each item belongs to a \textit{category}, and our goal is to compute a subset of items such that each category is \textit{fairly} represented, in addition to the total weight of the subset not exceeding the capacity, and the total value of the subset being maximized. We study various notions of \textit{group fairness}, such as the number of items from each category, the total value of items from each category, and the total weight of items from each category. \par

In recent years, a considerable amount of research \cite{TsangWRTZ19,JosephKMR16,Ali_junaid} has been focused on group fairness, i.e., to ensure that the algorithms are not  biased towards or against any  specific group in the population.  Here, the two key questions are: how to formalize the notion of fairness and how to design efficient algorithms that conform to these formalization.
One such formalization, disparate impact (DI) doctrine \cite{Chouldechova17,Chierichetti0LV17}, posits that any group must have approximately equal or proportional representation in the solution provided by the algorithm. Using this doctrine, Bera et al. \cite{BeraCFN19} introduced a notion of fairness in clustering problems where they deem a solution to be fair if it satisfies two properties: (a) {\em restricted dominance}, which upper bounds the fraction of items from a category, and (b) {\em minority protection}, which asserts a lower bound on the fraction of items from a category.  
They use LP based iterative rounding algorithm and the solution obtained by this algorithm might violate the {\em fairness} constraints by some additive amount. Similar group fairness notions as studied by Bera et al. \cite{BeraCFN19} for clustering problem, have little been studied for resource allocation problems \cite{Benabbou}. Resource allocation problems have been mostly studied under individual fairness and strategic viewpoint \cite{Segal-HaleviS18,BeiLMS19,Khamse-AshariLK18,Luss99,BenabbouCEZ19}. \par

These group fairness notions seem to arise naturally in many practical applications. One such scenario is the case of a server serving multiple clients. These clients can be viewed as categories. The clients have a set of jobs to be processed by the server. Each job has some specific resource requirement, which can be viewed as the weight of the item, and a parameter denoting importance, which can be viewed as the associated value. Since the server has limited amount of available resources, which can be viewed as the capacity of the knapsack, it has to select a subset of jobs that can be processed using available resources. It is expected that the selected subset is {\em fair} for each client.\par

The knapsack problem with fairness constraints can also be used to model the problem of \textit{fair} budget allocation in governing bodies. A government may have various project proposals related to different sectors such as agriculture, education, defense, etc. These sectors can be viewed as the categories. Each project has some cost and value (indicating social impact), which can be viewed as the associated weight and the value of an item,  respectively. A government wants to allocate its budget, which can be viewed as the capacity of the knapsack, such that each sector is \textit{fairly} addressed and the total social impact is maximized. The {\em fair max-knapsack} problem with group value constraint is well suited to model this scenario.

Similarly, we can also model the notion of \textit{participatory budgeting}, which has been proposed recently to take into account the preferences of various stakeholders in the organizations' budget \cite{Goel_A,Aziz_H,Shaprio}. In this process, each stakeholder provides a subset of projects, which can be viewed as items, it prefers. The cost of a project could be viewed as the weight of an item. The number of stakeholders preferring a project could be viewed as the value of an item. An organization could be further divided into subparts, which can be viewed as the categories. The stakeholders associated with a particular subpart would have a preference for specific projects. Any budgeting which does not consider such preferences might lead to an allocation that discriminates against some subparts of the organization. The {\em fair knapsack} problems described in this work are well suited to model this scenario. There are many other applications of our problems in advertising \cite{gollapudi2014fair}, web search \cite{inoue2016selecting}, etc. \par

If all items have identical weights, then the knapsack problem with value lower bound considered in this paper, is a special case of the committee election problem considered by Bredereck et al. \cite{Bredereck2017MultiwinnerEW}. Chekuri and Kumar \cite{Amit_Chekuri} studied maximum covering problems with group budget constraints. In this problem, we are given a collection of sets $\mathcal{S}=\{S_1,S_2, \dots, S_m\}$ where each $S_i$ is a subset of a given ground set $X$. We are also given a partition of $\mathcal{S}$ into $l$ groups. The goal is to pick $k$ sets from $\mathcal{S}$ such that at most one set is selected from each group and the cardinality of their union is maximized. Chekuri and Kumar \cite{Amit_Chekuri} show that the greedy algorithm gives constant factor approximation to this problem. This problem is analogous to one of our problem that has upper bound on number of items as fairness constraint.
\par
There are also some practical heuristics for the knapsack problems (for e.g. \cite{Keller,Meinolf_Sellmann}, etc.), but they do not consider {\em fairness} constraints. It is not straightforward to extend the techniques in these heuristics for knapsack problems with fairness constraints.

\subsection{Our contributions}
We study the following three notions of {\em group fairness} for the knapsack problem.

\begin{itemize}
	\item{\em Bound on the number of items. (\Cref{prob:number_max}) }Given a set of items, each belonging to  one of $\NoOfCategories$ categories, and a range for each category, the problem is to find a subset that maximizes the total value, such that the number of items from each category is in the given range, and the total weight of the subset is at most the capacity of the knapsack.

	\item{\em Bound on weight. (\Cref{prob:weight_max_knapsack}) }Given a set of items, each belonging to one of $\NoOfCategories$ categories, and upper and lower bounds for weight from each category, the problem is to find a subset that maximizes the total value, such that the total weight of items from each category is in between the given bounds, and the total weight of the subset is at most the capacity of the knapsack. 
	
	\item{\em Bound on value. (\Cref{prob:value_max_knapsack}) }Given a set of items, each belonging to one of $\NoOfCategories$ categories, and upper and lower bounds for value from each category, the problem is to find a subset that maximizes the total value, such that the total value of items from each category is in between the given bounds, and the total weight of the subset is at most the capacity of the knapsack.

\end{itemize}

We also study similar {\em fairness} notions in the min-knapsack problem which are defined analogously. We give approximation algorithms for these problems. For any input error parameter $\eps>0$, the running time of our algorithms are polynomial in the input size and $\frac{1}{\eps}$.
Our algorithms output a solution having the total value at least $(1-\eps)$ times the value of the optimal solution (for min-knapsack the weight is at most $(1+\eps)$ times optimal weight), but the solution produced by some of our algorithms might violate the fairness and/or the knapsack constraints by a small amount (multiplicative factor of $\pm \eps$). For these cases with violations, we show that it is even NP-hard to obtain a feasible solution without violating any constraint. Thus violations are necessary for these cases. We summarize the results of our algorithms in Table \ref{Final_res_table}. \par

\begin{table}
	\begin{tabular}{|M{2.8cm}|M{4cm}|M{4cm}|M{4.2cm}|}
		\hline
		\multicolumn{1}{|c|}{Knapsack Problem} & Number of items from category & Total value from category & Total weight from category \\ 
		\hline
		\multirow{2}{*}{Max-Knapsack}   & $(1-\eps,0,0)$    & $(1-\eps,\eps,0)$ & $(1,\eps,\eps)$\\ 
		 & $Theorem$ $\ref{the:alg_number_max}$ & $Theorem$ $\ref{the:alg_value_max_knapsack}$ & $Theorem$  $\ref{the:weight_max_knapsack}$\\ 
		[5pt]\hline
		\multirow{2}{*}{Min-Knapsack}  & $(1+\eps,0,0)$   & $(1,\eps,\eps)$&   $(1+\eps,\eps,0)$\\
		& $Theorem$ $\ref{the:number_min}$ &$Theorem$ $\ref{the:value_min}$ & $Theorem$ $\ref{the:weight_min}$\\ 
		\hline
	\end{tabular}
	\caption{Each entry in the table is represented by a triplet. The first entry in the triplet indicates the approximation ratio achieved by our algorithm. The second entry in the triplet indicates the fractional amount by which the output of the algorithm might violate the fairness constraints. The third entry in the triplet indicates the fractional amount by which the output of the algorithm might violate the knapsack constraint. Reference to the theorem describing respective algorithm is given below the triplet.}
	\label{Final_res_table}
\end{table}	

\paragraph{Proof overview:}
There is a dynamic programming (DP) based algorithm to solve the classical knapsack problem that runs in pseudo-polynomial time. By rounding the values (or weights) such that the rounded value belongs to a small set (if the problem has $n$ items to pack, then the value is rounded to the nearest multiple of $\eps/n$), the DP technique will give an ${\FPTAS}^{\ref{PTAS_def}}$ \cite{ibarra1975fast,lawler1979fast}. At a high level, we also use multi-level DP.  First, we compute bundles of different values of items from the same category. Next, we select one bundle from each category using a different DP table. \par

%% file: relatedwork.tex
\subsection{Related Work}\label{related_fairKnapsack}

\paragraph{Max knapsack problem:}
The classical knapsack problem was one of Karp's 21 \NP-complete problems. It is one of the fundamental problems in optimization and approximation algorithms.  $\FPTAS$ for the knapsack problem is known \cite{ibarra1975fast,lawler1979fast}. $\PTAS$\footnote{\label{PTAS_def}A polynomial-time approximation scheme $(\PTAS)$ is an algorithm which takes an instance of an optimization problem and a parameter $\eps > 0$ and, in polynomial time of input size for fixed $\eps$, produces a solution that is within a factor $1 + \eps$ of being optimal (or $1-\eps$ for maximization problems).\\
	A fully polynomial-time approximation scheme $(\FPTAS)$ is a $\PTAS$ with running time polynomial in both input size and $\frac{1}{\eps}$.

} for the multiple knapsack problem is also known \cite{Chekuri,Jansen}. The knapsack problem has also been studied under the multidimensional setting \cite{Frieze, GalvezGHI0W17}. Another variant of the knapsack problem is the submodular knapsack problem, where the value function is submodular \cite{Kulik,Sviridenko,Lee}. The knapsack problem has also been well-studied under an online setting \cite{Albers0L19}. For different variants of the knapsack and related bin packing problems, we refer the readers to \cite{ChristensenKPT17,khan2015approximation}.
\par

\paragraph{Min-knapsack problem: } The min-knapsack problem is a frequently encountered problem in the optimization. The problem admits an \FPTAS~\cite{csirik1991heuristics,padberg1985valid}. As the problem appears as a key substructure of numerous other optimization problems, the polyhedral study of this problem has led to the development of important tools, such as the knapsack cover inequalities and there is a rich connection of the problem with extension complexity  and sum-of-squares hierarchy \cite{BazziFHS17i,KurpiszLM14,FaenzaMMS18}. 

\paragraph{Class-constrained knapsack and fair knapsack: } To the best of our knowledge, the fairness notions for knapsack problems studied in this work have not been studied previously. The closest model to our model that has been studied before is the class-constrained multiple knapsack problem \cite{Shachnai}. This problem restricts the maximum number of classes from which items could be packed in a knapsack.  
The algorithm of the class constrained knapsack problem \cite{Shachnai} uses two levels of dynamic programming, and is similar in spirit to our algorithm

Fairness notions for knapsack under multi-agent valuations have also been studied \cite{Fluschnik}. 
They have several approaches to aggregate voters' preferences: individually best, diverse knapsack and fair knapsack. 
The main difference is that our model ensures fairness through constraints, while their model ensures fairness through objective function. The objective function used in their model is the Nash welfare function \cite{Nash}.
The diverse knapsack model in \cite{Fluschnik} only ensures one representative item from each category, whereas our model allows different lower and upper bounds for different groups. We  focus on approximation schemes for our problems, whereas  \cite{Fluschnik} focused on the parameterized complexity and the complexity of the problem under some restricted special domains.\par
Ajay et al.~\cite{Anand_Nimbodkar} study some notion of group fairness for online/offline matching problem. Their notions of fairness are similar to ours in some respect.

%% file: algorithm_max_knapsack.tex
\section{Preliminaries}\label{preliminaries}

\paragraph{Notations.}Let the set $\{1, 2, \dots, r\}$ be denoted by $[r]$. Let $\NoOfCategories$ denote the number of categories. The category number is denoted from the set $[\NoOfCategories]$. Each input item belongs to some category. Let $V_i$ denote the set of all items from category $i$, $\forall i \in [\NoOfCategories]$. For $i\in [\NoOfCategories]$, if $V_i$ has $k$ items, then $V_i:=[k]$, i.e. an item in $V_i$ is represented by a number in $[k]$. Also, assume that $w_{j}^{(i)}$ and $v_{j}^{(i)}$ are the weight and the value of item $j\in V_i$, respectively, $\forall i\in [\NoOfCategories]$. Also let $\max \{ \mid V_i\mid \mid i \in [\NoOfCategories] \}=n$. Let $\NoOfItmes$ be the total number of items. By our notations, $\NoOfItmes \leq n\NoOfCategories$.\par

All the algorithms presented in this paper use dynamic programming tables. The entry in the dynamic programming table might represent the total weight or value of a subset of items. The respective subsets can be obtained by maintaining the required references in the tables. Details of it are trivial and hence they are not discussed in this paper.\par

We will use following abbreviations for the problems discussed in this chapter.

\paragraph{Max knapsack problems:} Abbreviations for max knapsack problems: \\
$BN^{\max}$ : Bound (upper and lower both) on number of items in max knapsack (\Cref{prob:number_max}).\\
$BV^{\max}$ : Bound (upper and lower both) on value in max knapsack (\Cref{prob:value_max_knapsack}).\\
$BW^{\max}$ : Bound (upper and lower both) on weight in max knapsack (\Cref{prob:weight_max_knapsack}).

\paragraph{Min knapsack problems:} Abbreviations for min knapsack problems: \\
$BN^{\min}$ : Bound (upper and lower both) on number of items in min knapsack (\Cref{prob:number_min}).\\
$BV^{\min}$ : Bound (upper and lower both) on value in min knapsack (\Cref{prob:value_min}).\\
$BW^{\min}$ : Bound (upper and lower both) on weight in min knapsack (\Cref{prob:weight_min}).\\

\begin{problem}[max-knapsack]\label{classical_knapsack}
Given a set of items, each with a weight and a value, along with a maximum weight capacity $W$, find a subset of items maximizing the total value, such that the total weight of the subset is at most $W$.
\end{problem}

\begin{problem}[min-knapsack]\label{classical_min_knapsack}
	Given a set of items, each with a weight and a value, along with a lower bound $M$, find a subset of items minimizing  the total weight, such that the total value of the subset is at least $M$.
\end{problem}

\begin{problem}[subset sum]\label{subsetSum2}
	Given a set $I$ of non-negative rationals such that $\sum_{a \in I}a =1$, find a subset $S \subseteq I$ such that $\sum_{a \in S}a =\frac{1}{2}$.
\end{problem}

All of the above problems are known $\NP$-hard problems.

\section{Algorithms for fair max-knapsack}\label{Fair_max_knapsack}
We study fair max-knapsack problems for various notions of fairness and give an algorithm for each of them. We consider various parameters for fairness such as the number of items from each category, bound on the total value of items, and the total weight of items from each category.

\subsection{Fairness based on number of items}

\begin{problem}[$BN^{\max}$]\label{prob:number_max}
Given a set of items, each belonging to  one of $\NoOfCategories$ categories, and numbers $l_{i}^{n}$ and $u_{i}^{n}$ for each category $i, \forall i \in [\NoOfCategories]$, the problem is to find a subset that maximizes the total value, such that the number of items from category $i$ is between $l_{i}^{n}$ and $u_{i}^{n}$, $\forall i\in [\NoOfCategories]$, and the total weight of the subset is at most the capacity of the knapsack $B$. 	
\end{problem}

\begin{theorem}\label{the:alg_number_max}
For any $\eps>0$, there exists a $(1-\eps)$-approximation algorithm for $BN^{\max}$ (Problem \ref{prob:number_max}) with running time $O\left(\frac{n \NoOfItmes^3}{{\eps}^2}\right)$.
\end{theorem}
\begin{proof}
The algorithm is described in Algorithm \ref{alg:fairness_based_number}.The algorithm starts by rounding the values of all items so that the rounded values lie in a small range (Step \ref{rounding_number_max}). Then it creates bundles having different cardinality and total rounded value from items in $V_i$, $\forall i \in [\NoOfCategories]$. It uses the dynamic programming table $\mathcal{A}$ for this (Step \ref{A_number_max}). After that the algorithm combines bundles from all categories to obtain the final solution using the dynamic programming table $\mathcal{B}$ (Step \ref{B_number_max}). The following describes formal proof of the algorithm.\par	
	
	\RestyleAlgo{boxruled}
	\begin{algorithm}
		\caption{Algorithm for $BN^{\max}$ (Problem \ref{prob:number_max})} \label{alg:fairness_based_number}
		\textbf{Input:} The sets $V_{i}$ of items, $0\leq l_{i}^{n}\leq u_{i}^{n}$ for $i\in[\NoOfCategories]$, $B$ the capacity of knapsack and $\eps>0$. \\
		\textbf{Output:}  $S$ having the total value at least $1-\eps$ times the optimal value of $BN^{\max}$ (Problem \ref{prob:number_max}), such that $l_{i}^{n} \leq \mid S \cap V_i\mid \leq u_{i}^{n}$, $\forall i \in [\NoOfCategories]$ and the total weight of $S$ is at most the knapsack weight $B$.\\
		\begin{enumerate}
			
			\item \label{Remove item_number_max} Remove all items $j\in V_i$, $\forall i\in [\NoOfCategories]$ that do not come under any feasible solution. We can check whether an item comes under a feasible solution by checking the weight of least weight feasible subset containing the item is less than or equal to $B$. Let $v_{\max}$ be the maximum value of the remaining items.
			\item \label{rounding_number_max}$\forall i\in [\NoOfCategories]$, $\forall j \in  V_i$, let
			\[ {v}_{j}^{(i)}{'}:= \left \lfloor \frac{v_{j}^{(i)}}{\eps v_{\max}}\NoOfItmes \right\rfloor.\]
			
			\item \label{A_number_max} Let $\mathcal{A}(i,j,v,t)$, $\forall i \in [\NoOfCategories]$, $\forall j \in  V_i $, $\forall t \in V_i \cup \{0\}$, $\forall v \in \left[\left \lceil\frac{\NoOfItmes^2}{\eps}\right \rceil\right]\cup \{0\}$, be the dynamic \\programming table constructed in the following way,
			
			\begin{enumerate}

				\item \label{A1_number_max}$\forall i\in [\NoOfCategories]$,\\ $\mathcal{A}(i,1,{v}^{(i)}_{1}{'},1):=w^{(i)}_{1}$. $\mathcal{A}(i,1,0,0)=0$. $\mathcal{A}(i,1,{v{'}},.):=\infty$ for $v{'} \in \left[\left \lceil\frac{\NoOfItmes^2}{\eps}\right \rceil\right] \setminus \{v_{1}^{(i)'}\}$. 
				
				\item  \label{AI_number_max} $\forall i\in [\NoOfCategories]$ and $\forall j \in V_i\setminus\{1\}$, $\forall t \in V_i \cup \{0\}$, $\forall v \in \left[\left \lceil\frac{\NoOfItmes^2}{\eps}\right \rceil\right]$,
				
				If $v<v_{j}^{(i)}$ or $t=0$, then $\mathcal{A}(i,j,v,t)= \mathcal{A}(i,j-1,v,t) $.\\
				Else ,
			\begin{equation*}
							\begin{split}
				\mathcal{A}(i,j,v,t):=\min &\left \{ \mathcal{A}(i,j-1,v-{v}_{j}^{(i)}{'},t-1)+w_{j}^{(i)}, \right.\\ &\left.  \mathcal{A}(i,j-1,v,t) \mid 0 \leq {v}_{j}^{(i)}{'} \leq v  \right\}.
				\end{split}
				\end{equation*}

			\end{enumerate}

			\item \label{B_number_max}  Let $\mathcal{B}(i,v)$, $\forall i\in[\NoOfCategories]$, $\forall v \in \left[\left \lceil\frac{\NoOfItmes^2}{\eps}\right \rceil\right] \cup \{0\}$ be another the dynamic programming table.
			
			\begin{enumerate}

			\item \label{B1_number_max} $\forall v \in \left[\left \lceil\frac{\NoOfItmes^2}{\eps}\right \rceil\right] \cup \{0\}$ ,
			\[\mathcal{B}(1,v):=\min\{\mathcal{A}(1,\mid V_1\mid,v,t)\mid \; l_{1}^{n}\leq t\leq u_{1}^{n}\}.\] 
			
			\item\label{BI_number_max} For $i \in [\NoOfCategories]\setminus\{1\}$ and $v \in \left[\left \lceil\frac{\NoOfItmes^2}{\eps}\right \rceil\right] \cup \{0\}$,
			
			\begin{equation*}
			\begin{split}
			\mathcal{B}(i,v):= \min &\left\{ \mathcal{B}(i-1,v_r)+\mathcal{A}(i,\mid V_i \mid,v-v_r,t) \mid l_{i}^{n}\leq t\leq u_{i}^{n} \right. \\
			& \left. v_r \leq v \; \& \; v_{r} \in \left[\left \lceil\frac{\NoOfItmes^2}{\eps}\right \rceil\right] \cup \{0\} \right\} .
			\end{split}
			\end{equation*}
				
			\end{enumerate}

			\item \label{output_number_max}	Output $S$ in the following way,
			\[ \argmax_{v} \{ \mathcal{B}(\NoOfCategories,v) \mid \mathcal{B}(\NoOfCategories,v) \leq B \}. \]
			\end{enumerate}
\end{algorithm}

	\paragraph{Property of $\mathcal{A}$. }We claim that $\mathcal{A}(i,j,v,t)$, $\forall i \in [\NoOfCategories]$,$\forall t\in V_i \cup \{0\}$, $j\in V_i$, $\forall v \in \left[ \left \lceil  \frac{\NoOfItmes^2}{\eps} \right \rceil \right] \cup \{0\}$, denotes the weight of a minimum weight subset of cardinality $t$ from first $j$ items of $V_i$ having the total rounded value $v$. Let $S{'}$ be the subset satisfying above property for the entry $\mathcal{A}(i,j,v,t)$, $\forall i \in [\NoOfCategories]$, $\forall j \in V_i\setminus \{1\}$, $\forall t \in V_i\cup \{0\}$, $\forall v \in \left[ \left \lceil  \frac{\NoOfItmes^2}{\eps} \right \rceil \right] \cup \{0\}$.
	
	\begin{itemize}
	    \item If $j\in S{'}$, then $\mathcal{A}(i,j,v,t)$ is the sum of weight of $j$ ($w_{j}^{(i)}$) and the weight of minimum weight subset from first $j-1$ items of $V_i$ having cardinality $t-1$ and the rounded value $v-v_{j}^{(i)}{'}$ ($\mathcal{A}(i,j-1,v-v_{j}^{(i)}{'},t-1)$), which is equal to $\mathcal{A}(i,j-1,v-v_{j}^{(i)}{'},t-1)+w_{j}^{(i)}$.
	    \item If $j\notin S{'}$, then $\mathcal{A}(i,j,v,t)$ is equal to the weight of minimum weight subset from first $j-1$ items of $V_i$ having cardinality $t$ and the rounded value $v$, which is equal to $\mathcal{A}(i,j-1,v,t)$.  
	\end{itemize}
	 The recursion in Step \ref{AI_number_max} captures both of above possibilities. The Step \ref{A1_number_max} initializes base cases for the recursion in Step \ref{AI_number_max}. The entry of $\mathcal{A}$ also corresponds to the respective subset. The entry of $\mathcal{A}$ could be $\infty$, which indicates no subset. We use the notation $\mathcal{A}(i,j,v,t)$ to indicate both the entry and the subset.
	 
	\paragraph{Property of $\mathcal{B}$. }We claim that $\mathcal{B}(i,v)$, $\forall i \in [\NoOfCategories]$, $\forall v \in \left[ \left \lceil  \frac{\NoOfItmes^2}{\eps} \right \rceil \right]\cup \{0\}$, denotes the weight of a minimum weight subset of $\cup_{j=1}^{i} V_j$ having the total rounded value $v$, such that fairness constraints are satisfies for the subset for all the categories up to $i$. Let $S{'}$ be the subset of $\cup_{j=1}^{i} V_j$ satisfying above property for $\mathcal{B}(i,v)$, for some $ i \in [\NoOfCategories]\setminus\{1\}$ and $v \in \left[ \left \lceil  \frac{\NoOfItmes^2}{\eps} \right \rceil \right]\cup \{0\}$. If $\sum_{j\in S{'}\cap V_i}v_{j}^{(i)}{'}=v_r$, then $\mathcal{B}(i,v)$ is the sum of weights of minimum weight subset of $\cup_{j=1}^{i-1}V_j$ having the total rounded value $v-v_r$ that satisfies fairness condition for all categories up to $i-1$ ($\mathcal{B}(i-1,v-v_r)$), and the minimum weight subset of $V_i$ having total rounded value $v_r$ that satisfies fairness condition for category $i$ ($\mathcal{A}(i,\mid V_i\mid,v_r,t)$ such that $l_{i}^{n}\leq t \leq u_{i}^{n}$). The recursion in Step \ref{BI_number_max} captures this for all possible $v_r$. The entry of $\mathcal{B}$ also corresponds to the respective subset. The entry of $\mathcal{B}$ could be $\infty$, which indicates no subset. We use the notation $\mathcal{B}(i,v)$ to indicate both the entry and the subset.  \par

 By the property of $\mathcal{B}$, the total rounded value of $S$ (Step \ref{output_number_max}) is at least the total rounded value of any optimal solution. By \Cref{lema:number_max}, we can show that the total value of $S$ is at least $(1-\eps)$ times the optimal solution. 
	
	\paragraph{Running time analysis:}The size of the table $\mathcal{A}$ is $O\left(\frac{n\NoOfItmes^3}{\eps}\right)$, and to fill each entry in the table $\mathcal{A}$, we require $O(1)$ time. The size of the table $\mathcal{B}$ is $O\left(\frac{\NoOfItmes^2 \NoOfCategories}{\eps}\right)$, and to fill each entry in the table $\mathcal{B}$, we require $O\left(\frac{n\NoOfItmes^2}{\eps}\right)$ time. So, the total time required is $O\left(\frac{n\NoOfItmes^3}{\eps^2}\right)$.
\end{proof}

\begin{theorem}\label{lema:number_max}
	If $\OPT$ is the optimal objective value of $BN^{\max}$ (Problem \ref{prob:number_max}), then the total value of a set returned by Algorithm \ref{alg:fairness_based_number} is at least $(1-\eps)\OPT$.
\end{theorem}
\begin{proof}
	Let $O \subseteq \cup_{i=1}^{\NoOfCategories} V_{i}$ be the set of items in an optimal solution and $S \subseteq \cup_{i=1}^{\NoOfCategories} V_{i}$ be the set of items selected by Algorithm \ref{alg:fairness_based_number}. Since $S$ is an optimal solution of rounded value in Algorithm \ref{alg:fairness_based_number}, the total rounded value of $S$ is more than the total rounded value of $O$.  
	\begin{equation}\label{eq:optimal_rounded_sol_max}
	\sum_{i=1}^{\NoOfCategories} \sum_{j \in V_i \cap O} {v}_{j}^{(i)}{'} \leq \sum_{i=1}^{\NoOfCategories} \sum_{j \in V_i \cap S} {v}_{j}^{(i)}{'}.
	\end{equation}

	Because of the rounding in Step \ref{rounding_number_max}, we have following inequalities $\forall i\in [\NoOfCategories]$ and $\forall j\in  V_i $,
	\begin{equation}\label{eq:rounded_val_eq_max}
	\frac{\eps v_{\max}{v}_{j}^{(i)}{'}}{\NoOfItmes} \leq v_{j}^{(i)} \leq \frac{\eps v_{\max} ({v}_{j}^{(i)}{'}+1)}{\NoOfItmes} = \frac{\eps v_{\max} {v}_{j}^{(i)}{'}}{\NoOfItmes}+\frac{\eps v_{\max}}{\NoOfItmes}. 
	\end{equation}
	
	So we get,
	\begin{align*}
	\sum_{i=1}^{\NoOfCategories} \sum_{j \in V_i \cap O} v_{j}^{(i)} & \leq \frac{\eps v_{\max}}{\NoOfItmes} \left(\sum_{i=1}^{\NoOfCategories} \sum_{j\in V_i \cap O} ({v}_{j}^{(i)}{'}+1) \right) & \textrm{(using Inequality \ref{eq:rounded_val_eq_max})}\\
	& \leq \frac{\eps v_{\max}}{\NoOfItmes} \left(\sum_{i=1}^{\NoOfCategories} \sum_{j\in V_i \cap O} {v}_{j}^{(i)}{'} \right)+\eps v_{\max} & \textrm{$\left(\sum_{i=1}^{\NoOfCategories} \sum_{j\in V_i \cap O}1\leq \NoOfItmes\right)$}\\
			& \leq \frac{\eps v_{\max}}{\NoOfItmes} \left(\sum_{i=1}^{\NoOfCategories}  \sum_{j\in V_i \cap S} {v}_{j}^{(i)}{'} \right)+ \eps v_{\max} &\textrm{(using Inequality \ref{eq:optimal_rounded_sol_max})}\\
    & \leq \left(\sum_{i=1}^{\NoOfCategories}  \sum_{j\in V_i \cap S} v_{j}^{(i)} \right)+ \eps v_{\max}.&\textrm{(using Inequality \ref{eq:rounded_val_eq_max})}
	\end{align*}
	
Since Step \ref{Remove item_number_max} of the Algorithm \ref{alg:fairness_based_number} discards all the items which doesn't come in any feasible solution, we know that $v_{\max}\leq \OPT$. So,
	\[ \left(\sum_{i=1}^{\NoOfCategories} \sum_{j\in V_i \cap S} v_{j}^{(i)} \right) \geq (1-\eps)\OPT .\]
\end{proof}

\subsection{Fairness based on bound on value}\label{sec_lowerBound_VALMaxknapsack}
\begin{problem}[$BV^{\max}$]\label{prob:value_max_knapsack}
	Given a set of items $V$, each belonging to one of $\NoOfCategories$ categories, a lower bound $l_{i}^{v}\geq 0$ and an upper bound $u_{i}^{v}$, such that $u_{i}^{v}\geq l_{i}^{v},$ for each category $i$, the goal is to find a subset that maximizes the total value, such that the total value of items from the category $i$ is in between $l_{i}^{v}$ and $u_{i}^{v}$, $\forall i\in[\NoOfCategories]$, and the total weight of the subset is at most the knapsack capacity $B$. 	
\end{problem}

We prove that it is $\NP$-hard to obtain the feasible solution of an instance of $BV^{\max}$ (Problem \ref{prob:value_max_knapsack}). 

\begin{theorem}\label{hardness_result_lowebound_value}
	There is no polynomial time algorithm that outputs feasible solution of $BV^{\max}$ (Problem \ref{prob:value_max_knapsack}), assuming $\PP\neq \NP$.
\end{theorem}
\begin{proof}
Given an instance of subset sum (\Cref{subsetSum2}) with the  set $I$, we will construct an instance of $BV^{\max}$ (Problem \ref{prob:value_max_knapsack}). Let $\NoOfCategories=1$, $l_{1}^{v}=\frac{1}{2}$, $u_{1}^{v}=1$, $B=\frac{1}{2}$. $V_1$ is the set of items in the instance of Problem \ref{prob:value_max_knapsack}.  The items in $V_1$ correspond to the numbers in $I$. An item corresponding to some $a\in I$ has a value and a weight equal to $a$. This proves that if we have an algorithm for $BV^{\max}$ (\Cref{prob:value_max_knapsack}) that doesn't violate the fairness constraint, then we can solve subset sum. But assuming $\PP\neq \NP$, this is not possible.
\end{proof} 

Theorem \ref{hardness_result_lowebound_value} implies that there does not exists polynomial time algorithm for $BV^{\max}$ (Problem \ref{prob:value_max_knapsack}). We give here an algorithm for $BV^{\max}$ (\Cref{prob:value_max_knapsack}) that might violates fairness constraints for a category by small amount.
\begin{theorem}\label{the:alg_value_max_knapsack}
	For any $\eps>0$, there exists an algorithm for $BV^{\max}$ (\Cref{prob:value_max_knapsack}) that outputs a set $S$ having the total value at least $1-\eps$ times the optimal value of $BV^{\max}$ (\Cref{prob:value_max_knapsack}), such that $(1-\eps)l_{k}^{v} \leq \sum_{r \in S \cap V_k} v^{(k)}_r \leq (1+\eps)u_{k}^{v}$, $\forall k \in [\NoOfCategories]$, and the total weight of $S$ is at most $B$. The running time of the algorithm is $O \left( \frac{n^2\NoOfCategories \log^3 \NoOfCategories \log_{1+\eps}^{3} \left(\frac{\NoOfItmes v_{\max}}{v_{\min}}\right)}{\eps}\right)$, where
	$v_{\min}:=\min \left\{ v_{j}^{(i)} \mid i \in [\NoOfCategories] \; \& \; j \in  V_i \; \& \; v_{j}^{(i)}>0  \right\}$. and $v_{\max}:=\max$ $ \left\{ v^{(i)}_{j} \mid i \in [\NoOfCategories] \; \& \; j \in V_i   \right\}$.
\end{theorem}

Towards proving this theorem, we first study the following sub-problem.

\begin{problem}\label{prob:sub_value_max_knapsack}
	Given $v>0$, $\eps>0$ and a set $V{'} = [n]$ of items, with item $i \in V{'}$ having the value ${v}_{i}{'}$ and the weight ${w}_{i}{'}$, compute a subset $S \subseteq V{'}$ minimizing $\sum_{i \in S} {w}_{i}{'}$, such that $(1-\eps)v \leq \sum_{i \in S} {v}_{i}{'} \leq (1+\eps)v$.
\end{problem}
Problem \ref{prob:sub_value_max_knapsack} looks similar to the min-knapsack problem but it is different from it in the following way. The total value of output of min-knapsack problem is required to be more than the given lower bound, while the total value of the output of Problem \ref{prob:sub_value_max_knapsack} is required to be lying in a small range. We will use  Algorithm \ref{alg:minimum_weight_given_v} for Problem \ref{prob:sub_value_max_knapsack} to obtain bundles of items in $V_i$, $\forall i \in [\NoOfCategories]$, such that the total value of different bundles are in different required ranges.
\begin{theorem}\label{the:sub_value_max_knapsack}
	For any $\eps>0$ and $v>0$, there exists an algorithm for Problem \ref{prob:sub_value_max_knapsack} that outputs a subset $S$ having the total weight at most the optimal weight of Problem \ref{prob:sub_value_max_knapsack}, and $(1-3\eps)v \leq \sum_{i \in S} {v}_{i}{'} \leq (1+3\eps)v$. The running time of the algorithm is $O\left(\frac{n^2}{\eps} \right)$ \end{theorem}
\begin{proof}
The algorithm for the theorem is described in Algorithm \ref{alg:minimum_weight_given_v}. We give the proof of correctness of the algorithm below.\par

	\RestyleAlgo{boxruled}
	\begin{algorithm}[!t]
		\caption{Algorithm for Problem \ref{prob:sub_value_max_knapsack}} \label{alg:minimum_weight_given_v}
		\textbf{ Input:} $v>0$, $\eps>0$, a set $V{'} = [n]$ of items, with item $i \in V{'}$ having value ${v}_{i}{'}$ and weight ${w}_{i}{'}$.\\
		\textbf{ Output:}  A subset $S \subseteq V{'}$ having total weight at most the optimal weight for Problem \ref{prob:sub_value_max_knapsack}, such that $(1-4\eps)v \leq \sum_{i \in S} {v}_{i}{'} \leq (1+4\eps)v$.\\
		\begin{enumerate}
				\item 	Remove all items from $V{'}$ having value more than $(1+\eps)v$.
				\item \label{rounding_sub_value_max} For each item $i \in V{'}$,
			let $ {v}_{i}{''} := \left \lfloor \frac{n {v}_{i}{'}}{\eps v}  \right \rfloor.$
				\item Let $\mathcal{H}(i,v{''})$ for $ v{''} \in \left[ \left \lceil \frac{(1+2\eps)n}{\eps} \right \rceil \right]\cup \{0\}$ and $i \in [n]$, be the dynamic programming \\table constructed in the following way.
				
				\begin{enumerate}
				
				\item\label{lower_bound_val_H_maxknapsack_H1} $\mathcal{H}(1,{v}_{1}{''})={w}_{1}{'}$. $\mathcal{H}(1,0)=0$. $\mathcal{H}(1,{v{''}})=\infty$, $\forall v{''} \in \left[ \left \lceil \frac{(1+2\eps)n}{\eps} \right \rceil \right]\setminus \{ {v}_{1}{''}\}$. 
				\item \label{lower_bound_val_H_maxknapsack_H}
			 $\forall i\in[n]\setminus\{1\}$, $\forall v{''} \in \left[ \left \lceil \frac{(1+2\eps)n}{\eps} \right \rceil \right] $,\\ If $v{''}< v_{i}^{''}$, then $\mathcal{H}(i,v{''})=\mathcal{H}(i-1,v{''})$. \\
			Else,
			\[\mathcal{H}(i,v{''}) := \min\{\mathcal{H}(i-1,v{''}),\mathcal{H}(i-1,v{''}-v_{i}{''})+{w}_{i}{'} \; \mid \; {v}_{i}{''}\leq v{''}  \}. 
			\]			
				\end{enumerate}
				
				\item \label{lower_bound_val_output_step_subalg_max_knapsack} Output the subset $S$ as follows,
				\[ \min \left \{ \mathcal{H}(n,v{''}) \mid \left \lfloor \frac{(1-2\eps)n}{\eps} \right \rfloor \leq v{''} \leq \left \lfloor \frac{(1+2\eps)n}{\eps} \right \rfloor  \right\} . \]
			
			\end{enumerate}
			
	\end{algorithm}
	We claim that the entry $\mathcal{H}(i,v{''})$,$\forall i\in [n]$, $\forall v{''} \in \left[ \left \lceil \frac{(1+2\eps)n}{\eps} \right \rceil \right]\cup \{0\}$, indicates the weight of the least weight subset from the first $i$ items of $V{'}$ having the total rounded value equal to $v{''}$. Let $S{'}$ denote the subset satisfying above property for the entry $\mathcal{H}(i,v{''})$,$\forall i\in [n]\setminus\{1\}$, $\forall v{''} \in \left[ \left \lceil \frac{(1+2\eps)n}{\eps} \right \rceil \right]\cup \{0\}$. We can have two cases as below.
	\begin{itemize}
	    \item If $i\in S{'}$, then $\mathcal{H}(i,v{''})$ is equal to the sum of the weight of $i$ ($w_{i}{'}$), and the weight of minimum weight subset from first $i-1$ items having the total rounded value $v{''}-v_{i}{''}$ ($\mathcal{H}(i-1,v{''}-v_{i}{''})$).
	 \item If $i \notin S{'}$, then $\mathcal{H}(i,v{''})$ is equal to the weight of minimum weight subset from first $i-1$ items having the total rounded value $v{''}$, i.e. $\mathcal{H}(i,v{''})=\mathcal{H}(i-1,v{''})$.
	 \end{itemize}
	The recursion in Step \ref{lower_bound_val_H_maxknapsack_H} captures both of above possibilities. Step \ref{lower_bound_val_H_maxknapsack_H1} does necessary initialization for the recursion in Step \ref{lower_bound_val_H_maxknapsack_H}. Let $O$ denote the set of items in an optimal solution of Problem \ref{prob:sub_value_max_knapsack} and $S$ be the set output by Step \ref{lower_bound_val_output_step_subalg_max_knapsack} of Algorithm \ref{alg:minimum_weight_given_v}.  \begin{subequations}
	 \begin{align}
	    \left \lfloor \frac{(1-2\eps)n}{\eps} \right \rfloor  \leq \sum_{i \in S} {v}_{i}{''} & \leq \left \lfloor \frac{(1+2\eps)n}{\eps} \right \rfloor \\
	\Rightarrow \frac{(1-2\eps)n}{\eps}-1  \leq \sum_{i \in S} {v}_{i}{''} & \leq  \frac{(1+2\eps)n}{\eps}. \label{S_range}
	 \end{align}
	 \end{subequations}
	 Using the definition of ${v}_{i}{''}$ of Step \ref{rounding_sub_value_max}, we obtain from the left inequality of \eqref{S_range},
	  \begin{align*}
         (1-2\eps)v-\frac{\eps v}{n} & \leq \sum_{i \in S} {v}_{i}{'}\\ \Rightarrow (1-3\eps)v  & \leq \sum_{i \in S} {v}_{i}{'} .
    \end{align*}
    	 From right inequality of \eqref{S_range},
    \begin{align*}
        \sum_{i \in S} ({v}_{i}{''}+1) &\leq \frac{(1+2\eps)n}{\eps}+n \\
        \Rightarrow \sum_{i \in S} {v}_{i}{'} &\leq (1+3\eps)v. &\textrm{(Definition of ${v}_{i}{''}$)}
    \end{align*}

	Now, we will prove that $\left \lfloor \frac{(1-2\eps)n}{\eps} \right \rfloor  \leq \sum_{i \in O} {v}_{i}{''} \leq \left \lfloor \frac{(1+2\eps)n}{\eps} \right \rfloor$. Since $S$ is the least weight subset in the previous range, the above claim will imply that the total weight of $S$ is at most the total weight of $O$. We know that $(1-\eps)v  \leq \sum_{i \in O} {v}_{i}{'}  \leq  (1+\eps) v$. Using the definition of ${v}_{i}{''}$ in Step \ref{rounding_sub_value_max}, we get 
	\begin{align*}
	\frac{n(1-\eps)}{\eps} -n & \leq \sum_{i \in O} {v}_{i}{''}  \leq \frac{(1+\eps)n}{\eps}+n \\
	\frac{n(1-2\eps)}{\eps} & \leq \sum_{i \in O} {v}_{i}{''}  \leq \frac{(1+2\eps)n}{\eps}\\
	\left\lfloor\frac{n(1-2\eps)}{\eps}\right\rfloor & \leq \sum_{i \in O} {v}_{i}{''}  \leq \left\lfloor\frac{(1+2\eps)n}{\eps}\right\rfloor.
	\end{align*}
	The last inequality is true by the fact that ${v}_{i}{''}$ is an integer, $\forall i \in V{'}$.
	Note that the algorithm returns $\infty$ if it does not find a subset in Step \ref{lower_bound_val_output_step_subalg_max_knapsack}. This is because of the initialization in Step \ref{lower_bound_val_H_maxknapsack_H1}.
	
	\paragraph{Running time analysis.} The size of the table $\mathcal{H}$ is $O\left(\frac{n^2}{\eps}\right)$, and to fill each entry in the table $\mathcal{H}$, we require $O(1)$ time. So, the total running time of the algorithm is $O\left(\frac{n^2}{\eps}\right)$.
	
\end{proof}

\begin{proof}[Proof of Theorem \ref{the:alg_value_max_knapsack}]
The algorithm for theorem is described in Algorithm \ref{alg:fairness_based_value}.
The algorithm creates bundles of items from $V_i$, $\forall i \in [\NoOfCategories]$, such that the total value of each bundle is in different ranges using Theorem \ref{the:sub_value_max_knapsack} (Step \ref{W_max}). It stores these bundles in the table $\mathcal{W}$. After that the algorithm combines bundles from all categories in divide and conquer fashion to obtain the final solution using the dynamic programming table $\mathcal{X}$ (Step \ref{X_max}). The total value of each bundle is represented by some power of $\left(1+\eps{'}\right)$ in the tables $\mathcal{W}$ and $\mathcal{X}$. So, the algorithm might lose at most $\left(1+\eps{'} \right)$ fraction of total value in one iteration of Step \ref{X_max}. The total fraction of value lost in the calculation after combining the bundles from all the categories by divide and conquer fashion in Step \ref{X_i_max} is at most $\left(1+\eps{'} \right)^{O(\log_{2}\NoOfCategories)}$. This is at most $\left(1+\eps\right)$ because of the choice of $\eps{'}$ in Step \ref{round_eps_value_max}. We describe the formal proof of the algorithm below.\par
	
		\RestyleAlgo{boxruled}
	\begin{algorithm}[!t]
		\caption{Algorithm for $BV^{\max}$ (Problem \ref{prob:value_max_knapsack}) } \label{alg:fairness_based_value}
		\textbf{ Input:} The sets of items: $V_{i}$ for all $i\in [m]$, $0 \leq l_{i}^{v}\leq u_{i}^{v}$ for all $ i\in[\NoOfCategories]$, the knapsack capacity $B$ and $\eps>0$. \\
		\textbf{ Output:} $S$ having the total value at least $1-\eps$ times the optimal value of $BV^{\max}$ (Problem \ref{prob:value_max_knapsack}), such that $(1-\eps)l_{i}^{v} \leq \sum_{r \in S \cap V_i} v_{r}^{(i)}\leq (1+\eps)u_{i}^{v}$, $\forall i \in [\NoOfCategories]$, and the total weight of $S$ is at most $B$.\\
		\begin{enumerate}
			\item \label{round_eps_value_max}	Let $\eps{'}=\left(1+\frac{3\eps}{8}\right)^{\frac{1}{\log_{2}\NoOfCategories+1}}-1$.  Also let $W_{range}:=\left[ \left \lceil \log_{\left(1+\frac{\eps}{8}\right)} \left(\frac{nv_{\max}}{v_{\min}}\right) \right \rceil\right]\cup \{0\}$  and
 			$X_{range}:=\left[ \left \lceil \log_{1+\eps{'}} \left(\frac{\NoOfItmes v_{\max}}{v_{\min}}\right) \right \rceil \right] \cup \{-\log_{2}\NoOfCategories, \dots, -1,0\}$.

			\item \label{W_max} Let $\mathcal{W}(i,j), \forall i \in [\NoOfCategories]$, $\forall j \in W_{range}$, be the table, where the entry $\mathcal{W}(i,j)$ indicates the weight \\of a subset of $V_i$ that is obtained by Theorem \ref{the:sub_value_max_knapsack} by setting $V_i$ as $V{'}$, $v_{\min}\left(1+\frac{\eps}{8}\right)^j$ as $v$, and $\frac{\eps}{6}$ \\as $\eps$ in Theorem \ref{the:sub_value_max_knapsack}.
			
			\item \label{X_max}Let $\mathcal{X}(i,j)$, $\forall i\in [2\NoOfCategories-1]$, $\forall j \in X_{range}$, be the DP table constructed as follows.
			
			\begin{enumerate}
			\item \label{X_1_max}
			$\forall i \in [\NoOfCategories]$, $\forall j\in X_{range}$,
			\begin{equation*}
			\begin{split}
			\mathcal{X}(\NoOfCategories-1+i,j):=\min & \left \{\mathcal{W}(i,j{''}) \mid  \;
		   \left(1+\frac{\eps}{8}\right)^{j{''}}\geq \left(1+\eps{'}\right)^{j}  \; \& \; v_{\min}\left(1+\frac{\eps}{8}\right)^{j{''}+1}   \right.\\ &\left. \geq l_{i}^{v} \& \; v_{\min}\left(1+\frac{\eps}{8}\right)^{j{''}-1}\leq u_{i}^{v}\; \&  \;  j{''} \in W_{range}   \right\} .
			\end{split}
			\end{equation*}
			
			 If the set satisfying above condition is empty, then set $\mathcal{X}(i,j)$ to $\infty$.
			\item \label{X_i_max} $\forall i\in[\NoOfCategories-1]$, $\forall j{'},j{''} \in X_{range}$, we have
			
			\begin{equation*}
			\begin{split}
			\mathcal{X}(i,j) := \min & \left\{ \mathcal{X}(2i,j{'})+\mathcal{X}(2i+1,j{''}) \mid \; \left(1+\eps{'}\right)^{j} \leq \left(1+\eps{'}\right)^{j{'}}+\left(1+\eps{'}\right)^{j{''}} \right\} .
			\end{split}
			\end{equation*}
			If the set satisfying above condition is empty, then set $\mathcal{X}(i,j)$ to $\infty$.
			
			\end{enumerate}

		\item \label{Output_lower_value_max} Output the subset $S$ in the following way, 
		\[ \argmax_{j} \{ \mathcal{X}(1,j) \mid \mathcal{X}(1,j) \leq B \} . \]	
		\end{enumerate}
		
	\end{algorithm}

\paragraph{Properties of $\mathcal{W}$.}We claim that the entry $\mathcal{W}(i,j)$,  $\forall i\in[\NoOfCategories]$, $j \in W_{range}$, indicates the weight of the subset of $V_i$ that satisfies the two properties listed below. The entry of $\mathcal{W}$ also corresponds to respective subset. The entry of $\mathcal{W}$ could be $\infty$, which indicates no subset. We use the notation $\mathcal{W}(i,j)$ to indicate both the entry and the subset. 
\begin{enumerate}
	\item \label{W_lower_value_max_1} If the entry $\mathcal{W}(i,j)$ is finite, then the total value of the corresponding subset is in between \\ $\left(1-\frac{\eps}{2}\right)\left(1+\frac{\eps}{8}\right)^j v_{\min}$ and $\left(1+\frac{\eps}{2}\right)\left(1+\frac{\eps}{8}\right)^{j}v_{\min}$.
	\item \label{W_lower_value_max_2} The total weight of the subset corresponding to $\mathcal{W}(i,j)$ is at most the total weight of any subset of $V_i$ having the total value in between $v_{\min}\left(1-\frac{\eps}{6}\right)\left(1+\frac{\eps}{8}\right)^{j}$ and $v_{\min}\left(1+\frac{\eps}{6}\right)$ $\left(1+\frac{\eps}{8}\right)^{j}$ (Since $\frac{\eps}{8}< \frac{\eps}{6}$, the total weight of $\mathcal{W}(i,j)$ will be less than or equal to the total weight of any subset  of $V_i$ having the total value in between $\left(1+\frac{\eps}{8}\right)^j v_{\min}$ and $\left(1+\frac{\eps}{8}\right)^{j+1}v_{\min}$).
\end{enumerate}

The table $\mathcal{W}$ is created in Step \ref{W_max} of the algorithm. This step uses Theorem \ref{the:sub_value_max_knapsack}. We get both above properties because of the guarantee of Theorem \ref{the:sub_value_max_knapsack}.

Let $\mathcal{T}$ be a perfect binary tree with $\NoOfCategories$ leaf nodes. For simplicity, assume that $\NoOfCategories$ is power of $2$. Although, we can prove the same result by slight modification of the proof when $\NoOfCategories$ is not power of $2$. The total number of nodes in $\mathcal{T}$ will be $2\NoOfCategories-1$.  Each node in $\mathcal{T}$ could be represented by an index number from $1$ to $2\NoOfCategories-1$ with root at index $1$.
The node at an index $i$ has an left child at an index $2i$ and right child at an index $2i+1$, $\forall i\in[\NoOfCategories-1]$. Let the leaf node at an index $(\NoOfCategories-1)+i$ represent the category $i$, $\forall i\in [\NoOfCategories]$. Let $\mathcal{T}(i)$ denote the set of categories represented by the leaves of sub tree rooted at node $i$. Specifically, $\mathcal{T}(\NoOfCategories-1+i)=\{i\}$, $\forall i\in[m]$. Also, $\mathcal{T}(i)=\mathcal{T}(2i)\cup\mathcal{T}(2i+1)$, $\forall i\in [\NoOfCategories-1]$.
\paragraph{Properties of $\mathcal{X}$.}We claim that the entry $\mathcal{X}(i,j)$,  $\forall i\in[2\NoOfCategories-1]$, $\forall j \in X_{range}$, indicates the weight of a subset of $\cup_{k\in \mathcal{T}(i)} V_{k}$ that satisfies three properties mentioned below. The entry of $\mathcal{X}$ also corresponds to the  respective subset. The entry of $\mathcal{X}$ could be $\infty$, which indicates no subset. We use the notation $\mathcal{X}(i,j)$ to indicate both the entry and the subset. 
\begin{enumerate}
	\item \label{X_lower_value_max_first} If the entry $\mathcal{X}(i,j)$ is finite, then
	\[\sum_{k\in \mathcal{T}(i)} \sum_{r \in \mathcal{X}(i,j)\cap V_{k}} v^{(k)}_{r} \geq \left(1-\frac{\eps}{2}\right)\left(1+\eps{'}\right)^{j}v_{\min}\].
	\item \label{X_lower_value_max_second}
	If the entry $\mathcal{X}(i,j)$ is finite, then $\forall k \in \mathcal{T}(i)$, \[\left(1-\frac{\eps}{2}\right)\left(1+\frac{\eps}{8}\right)^{-1}l_{k}^{v}\leq  \sum_{r \in  \mathcal{X}(i,j)\cap V_k} v_{r}^{(k)}  \leq \left(1+\frac{\eps}{2}\right)\left(1+\frac{\eps}{8}\right)u_{k}^{v}.\]
	
	\item \label{X_lower_value_max_third} Let $i$ be a node of $\mathcal{T}$, $\forall i \in [2\NoOfCategories-1]$, having the distance $t$ from leaves, $ t\in\left[\log_{2}\NoOfCategories\right]\cup \{0\}$.  For all $O{'} \subseteq \cup_{k\in \mathcal{T}(i)} V_{k}$ having the total value at least $\left(1+\eps{'}\right)^{j+t}v_{\min}$, and $ l_{k}^{v} \leq \sum_{r \in  O{'}\cap V_{k}}v_{r}^{(k)}\leq u_{k}^{v}$, $\forall k \in \mathcal{T}(i)$, the total weight of the subset  $\mathcal{X}(i,j)$ is at most the total weight of $O{'}$. 
\end{enumerate}
In Property \ref{X_lower_value_max_third}, any considered set $O{'}$ will have value at least $v_{\min}$. Thus in Property \ref{X_lower_value_max_third}, $t+j \ge 0$, i.e. $j \ge -\log_{2}\NoOfCategories$.
Hence, the minimum value in $X_{range}$ has been set to $-\log_{2}\NoOfCategories$ (Step \ref{round_eps_value_max}).\par 
Steps \ref{X_1_max} of the algorithm chooses the subset $\mathcal{W}(i,j{''})$ that satisfies the inequalities $\left(1+\frac{\eps}{8}\right)^{j{''}+1}v_{\min}\geq l_{i}^{v}$ and $\left(1+\frac{\eps}{8}\right)^{j{''}-1}v_{\min}\leq u_{i}^{v}$. By Property \ref{W_lower_value_max_1} of $\mathcal{W}$, the total value of $\mathcal{W}(i,j{''})$ is in between $\left(1-\frac{\eps}{2}\right)$ $\left(1+\frac{\eps}{8}\right)^{-1}l_{i}^{v}$ and $\left(1+\frac{\eps}{2}\right)\left(1+\frac{\eps}{8}\right)u_{i}^{v}$. This proves that the subset corresponding to finite entry $\mathcal{X}(i,j)$, $\forall i\in [2\NoOfCategories-1]$, $\forall  j \in X_{range}$, satisfies Property \ref{X_lower_value_max_second}. The subset corresponding to finite entry $\mathcal{X}(\NoOfCategories-1+i,j)$, $\forall i \in[\NoOfCategories]$ $\forall j \in  X_{range}$, will satisfy Property \ref{X_lower_value_max_first}. This is true because of Property \ref{W_lower_value_max_1} of $\mathcal{W}$ and the condition $\left(1+\frac{\eps}{8}\right)^{j{''}}\geq \left(1+\eps{'}\right)^{j}$ in Step \ref{X_1_max} for selecting subset $\mathcal{W}(i,j{''})$. Because of the condition $\left(1+\eps{'}\right)^{j}\leq (1+\eps{'})^{j{'}}+ (1+\eps{'})^{j{''}}$ in Step \ref{X_i_max} of the algorithm, the subset corresponding to finite entry $\mathcal{X}(i,j)$, $\forall i\in [\NoOfCategories-1]$, $\forall j \in X_{range}$, will satisfy Property \ref{X_lower_value_max_first}. \par

We prove that the nodes in $\mathcal{T}$ will satisfy Property \ref{X_lower_value_max_third} by induction. In the base case, we prove that Property \ref{X_lower_value_max_third} is satisfied for all leaves. Let $O{''}$ be any subset of $V_i$ that satisfies the fairness bounds such that the total value of $O{''}$ is at least $v_{\min}\left(1+\eps{'} \right)^{j}$. Let $j{''}\in W_{range}$ such that the total value of $O{''}$ is in between $v_{\min}\left(1+\frac{\eps}{8}\right)^{j{''}-1}$ and $v_{\min}\left(1+\frac{\eps}{8}\right)^{j{''}}$. By Property \ref{W_lower_value_max_2} of $\mathcal{W}$, the total weight of $\mathcal{W}(i,j{''})$ is at most the total weight of $O{''}$. Since $O{''}$ satisfies the fairness bounds, the conditions $v_{\min}\left(1+\frac{\eps}{8}\right)^{j{''}+1}\geq l_{i}^{v}$ and $v_{\min}\left(1+\frac{\eps}{8}\right)^{j{''}-1}\leq u_{i}^{v}$ in Step \ref{X_1_max} are also satisfied. So, the subset $\mathcal{W}(i,j{''})$ is feasible for $\mathcal{X}(\NoOfCategories-1+i,j)$ in Step \ref{X_1_max}. So, the Property \ref{X_lower_value_max_third} is satisfied by $\mathcal{X}(\NoOfCategories-1+i,j)$.\par

For any $t\in\left[\log_{2}\NoOfCategories-1\right]\cup \{0\}$, assume the hypothesis that all the nodes having distance $t$ from leaves satisfy Property \ref{X_lower_value_max_third}. We will prove by induction that for any node $i$ with distance $t+1$ from leaves, Property \ref{X_lower_value_max_third} is satisfied. Let $O{'}\subseteq \cup_{k\in \mathcal{T}(i)} V_k$  satisfies the fairness conditions for all categories in $\mathcal{T}(i)$.
Let $j{^*} \in X_{range}$ satisfies the following inequality

\begin{equation}\label{eq:bound_value_max_knapsack_1}
\left(1+\eps{'}\right)^{j{^*}}v_{\min} \leq \sum_{k\in \mathcal{T}(2i)}\sum_{r \in O{'}\cap V_k}v^{(k)}_{r} \leq\left(1+\eps{'}\right)^{j{^*}+1}v_{\min}.    
\end{equation}
If $j'=j{^*}-t$,
the induction hypothesis implies the following :
\begin{equation}\label{eq:weight_lower_bound_left_child}
\sum_{k\in \mathcal{T}(2i)} \sum_{r \in \mathcal{X}(2i,j{'})\cap V_{k}} w^{(k)}_{r} \leq \sum_{k\in \mathcal{T}(2i)} \sum_{r \in O{'} \cap V_{k}} w^{(k)}_{r}.
\end{equation}
Similarly, let $j{^{**}} \in X_{range}$ that satisfies the following inequality 

\begin{equation}\label{eq:bound_value_max_knapsack_2}
\left(1+\eps{'}\right)^{j{^{**}}}v_{\min} \leq \sum_{k\in \mathcal{T}(2i+1)}\sum_{r \in O{'}\cap V_k}v^{(k)}_{r} \leq\left(1+\eps{'}\right)^{j{^{**}}+1}v_{\min}.    
\end{equation}
If $j''=j{^{**}}-t$, the induction hypothesis implies the following \begin{equation}\label{eq:weight_lower_bound_right_child}
\sum_{k\in \mathcal{T}(2i+1)} \sum_{r \in \mathcal{X}(2i+1,j{''})\cap V_{k}} w^{(k)}_{r} \leq \sum_{k\in \mathcal{T}(2i+1)} \sum_{r \in O{'} \cap V_{k}} w^{(k)}_{r}.    
\end{equation}

We claim that the inequality $\left(1+\eps{'}\right)^{j}\leq \left(1+\eps{'}\right)^{j{'}}+ \left(1+\eps{'}\right)^{j{''}} $ is satisfied in Step \ref{X_i_max} for any $j\in X_{range}$, such that the total value of $O{'}$ is at least $v_{\min}(1+\eps{'})^{j+t+1}$. This is true because the maximum value of $O{'}$ is at most $v_{\min}\left(\left(1+\eps{'}\right)^{j{^*}+1}+\left(1+\eps{'}\right)^{j{^{**}}+1} \right)$ (by Inequality \ref{eq:bound_value_max_knapsack_1} and Inequality \ref{eq:bound_value_max_knapsack_2}), which is more than $v_{\min}\left(1+\eps{'}\right)^{j+t+1}$. So, the pair of subsets $\mathcal{X}(2i,j{'})$ and $\mathcal{X}(2i+1,j{''})$ is feasible in Step \ref{X_i_max} for all such $j$. By Equation \ref{eq:weight_lower_bound_left_child} and Equation \ref{eq:weight_lower_bound_right_child}, Property \ref{X_lower_value_max_third} of $\mathcal{X}$ is satisfied for $\mathcal{X}(i,j)$ and $O{'}$.\par	

Let $\OPT$ be the optimal value. Let $j$ be the number in $X_{range}$ such that

\begin{equation}\label{eq:optimal_bound_value_max}
\left(1+\eps{'}\right)^{j+\log_{2} \NoOfCategories}v_{\min}\leq \OPT \leq \left(1+\eps{'}\right)^{j+\log_{2} \NoOfCategories+1}v_{\min}. 
\end{equation}
 If we apply Property \ref{X_lower_value_max_third} of $\mathcal{X}$ to any optimal solution, we get that the total weight of the subset $\mathcal{X}(1,j)$ is less than or equal to the total weight of an optimal subset. So, the subset $\mathcal{X}(1,j)$ is feasible in Step \ref{Output_lower_value_max}. By Property \ref{X_lower_value_max_first}, the total value of $\mathcal{X}(1,j)$ is at least $\left(1-\frac{\eps}{2}\right)\left(1+\eps{'}\right)^{j}v_{\min}$. By Inequality \ref{eq:optimal_bound_value_max}, this value is at least 
 \begin{align*}
\left(1-\frac{\eps}{2}\right)\frac{\OPT}{\left(1+\eps{'}\right)^{\log_{2} \NoOfCategories+1}}&= \left(1-\frac{\eps}{2}\right)\frac{\OPT}{\left(1+\frac{3\eps}{8}\right)}\\
&> \left(1-\frac{\eps}{2}\right)\left(1-\frac{3\eps}{8}\right)\OPT\\
&> (1-\eps)\OPT.     
 \end{align*}
If the entry $\mathcal{X}(1,j)$ is finite, Property \ref{X_lower_value_max_second} of $\mathcal{X}$ implies that the total value of items from $V_i \cap \mathcal{X}(1,j)$ is in between $\left(1-\frac{\eps}{2} \right)\left(1+\frac{\eps}{8} \right)^{-1}l_{i}^{v}$ and $\left(1+\frac{\eps}{2} \right)\left(1+\frac{\eps}{8} \right)u_{i}^{v}$ for all $i\in [\NoOfCategories]$. The total value of items from $V_i \cap \mathcal{X}(1,j)$ is at least $\left(1-\frac{\eps}{2} \right)\left(1+\frac{\eps}{8} \right)^{-1}l_{i}^{v}> \left(1-\frac{\eps}{2} \right)\left(1-\frac{\eps}{8} \right) l_{i}^{v}> (1-\eps)l_{i}^{v}$. The total value of items from $V_i \cap \mathcal{X}(1,j)$ is at most $\left(1+\frac{\eps}{2} \right)\left(1+\frac{\eps}{8} \right)u_{i}^{v} < \left(1+\frac{6}{8}\eps \right)u_{i}^{v}<\left(1+\eps \right)u_{i}^{v}$. \par

\paragraph{Running time analysis:}  The size of the $\mathcal{W}$ is $O \left( \NoOfCategories \log_{1+\eps} \left(\frac{\NoOfItmes v_{\max}}{v_{\min}}\right)\right)$. We require $O\left(\frac{n^2 }{\eps}\right)$ time to fill each entry of the $\mathcal{W}$. So, the total time required to build the table $\mathcal{W}$ is $O \left( \frac{n^2\NoOfCategories \log_{1+\eps} \left(\frac{\NoOfItmes v_{\max}}{v_{\min}}\right)}{\eps}\right)$. The total time required to build the table $\mathcal{X}$ is $O\left(\NoOfCategories \log_{1+\eps{'}}^{3} \left(\frac{\NoOfItmes v_{\max}}{v_{\min}}\right)\right)=O\left(\NoOfCategories \log^3 \NoOfCategories \log_{1+\eps}^{3} \left(\frac{\NoOfItmes v_{\max}}{v_{\min}}\right)\right)$. So, the total running time of the algorithm is $O \left( \frac{n^2\NoOfCategories \log^3 \NoOfCategories \log_{1+\eps}^{3} \left(\frac{\NoOfItmes v_{\max}}{v_{\min}}\right)}{\eps}\right)$. 
\end{proof}

\subsection{Fairness based on bound on weight}\label{sec_lower_Weight_max}
\begin{problem}[$BW^{\max}$]\label{prob:weight_max_knapsack}
Given a set of items, each belonging to one of $\NoOfCategories$ categories, a lower bound $l_{i}^{w}$ and an upper bound $u_{i}^{w}$ for each category $i$, $\forall i\in[\NoOfCategories]$, the goal is to find a subset that maximizes the total value, such that the total weight of items from category $i$ is in between $l_{i}^{w}$ and $u_{i}^{w}$, $\forall i \in [\NoOfCategories]$ and the total weight of the subset is at most the capacity of the knapsack $B$. 
\end{problem}
We prove that it is $\NP$-hard to obtain the feasible solution of an instance of $BW^{\max}$ (Problem \ref{prob:weight_max_knapsack}). 

\begin{theorem}\label{hardness_result_lowebound_weight}
	There is no polynomial time algorithm which outputs a feasible solution of $BW^{\max}$ (Problem \ref{prob:weight_max_knapsack}), assuming $\PP \neq \NP$.
\end{theorem}
\begin{proof}
	Given an instance of subset sum (\Cref{subsetSum2}), we can construct an instance of $BW^{\max}$ (Problem \ref{prob:weight_max_knapsack}) in following way. Let $\NoOfCategories=1$, $l_{1}^{w}=\frac{1}{2}$, $u_{1}^{w}=1$ and $B=\frac{1}{2}$. The set $V_1$ contains items which correspond to the numbers in $I$. An item corresponding to some $a\in I$ has a value and a weight equal to $a$. This proves that if we have an algorithm that outputs a feasible solution of $BW^{\max}$ (Problem \ref{prob:weight_max_knapsack}) without violating any constraint, then we can solve subset sum (Problem \ref{subsetSum2}). But assuming $\PP\neq \NP$, this is not possible.	
\end{proof} 

Theorem \ref{hardness_result_lowebound_weight} implies that there does not exists polynomial time algorithm for $BW^{\max}$ (\Cref{prob:weight_max_knapsack}). We give here an algorithm for $BW^{\max}$ (\Cref{prob:weight_max_knapsack}) that might violates fairness constraints and knapsack constraint by small amount.

\begin{theorem}\label{the:weight_max_knapsack}
	For any $\eps>0$, there exists an algorithm for $BW^{\max}$ (\Cref{prob:weight_max_knapsack}) that outputs a solution $S$ whose objective value is at least the optimal value, and $(1-\eps)l_{i}^{w} \leq \sum_{j \in S \cap V_i}w_{j}^{(i)} \leq (1+\eps)u_{i}^{w} $, and $\sum_{i=1}^{l} \sum_{j \in S \cap V_i} w_{j}^{(i)} \leq (1+\eps)B$, $\forall i \in [\NoOfCategories]$. The running time of the algorithm is $O \left( \frac{n^2\NoOfCategories \log^3 \NoOfCategories \log_{1+\eps}^{3} \left (\frac{B}{w_{\min}}\right)}{\eps}\right)$, where $w_{\min}:= \min \left\{w_{j}^{(i)} \mid i \in [\NoOfCategories] \; \& \; j \in  V_i \; \& \; w_{j}^{(i)}>0 \right\}$. 
\end{theorem}

Towards proving this theorem, we first study the following sub-problem.

\begin{problem}\label{prob:weight_sub_max_knapsack}
	Given $w>0$, $\eps>0$ and a set $V{'} = [n]$ of items, with item $i \in V{'}$ having the value ${v}_{i}{'}$ and the weight ${w}_{i}{'}$, compute a subset $S \subseteq V{'}$ maximizing $\sum_{i \in S} {v}_{i}{'}$, such that $(1-\eps)w \leq \sum_{i \in S} {w}_{i}{'} \leq (1+\eps)w$.
\end{problem}
Above problem looks similar to the max-knapsack problem but it is different from it in the following way. The total weight of output subset in the max-knapsack problem is required to be less than or equal to the given bound, while the total weight of output subset in \Cref{prob:weight_sub_max_knapsack} is required to be in small range. We could use the algorithm for \Cref{prob:weight_sub_max_knapsack} to obtain bundles of items $V_i$, $\forall i \in [\NoOfCategories]$ , such that the total weight of different bundles are in different ranges.
\begin{theorem}\label{the:sub_weight_max_knapsack}
	For any $\eps>0$ and $w>0$, there exists an algorithm which outputs a subset $S$ having the total value at least the optimal value of Problem \ref{prob:weight_sub_max_knapsack}, and $(1-3\eps)w \leq \sum_{i \in S} w_{i}{'} \leq (1+3\eps)w$. The running time of the algorithm is $O\left(\frac{n^2}{\eps}\right)$.
\end{theorem}
\begin{proof}
	
	The algorithm for the theorem is described in Algorithm \ref{alg:maximum_value_given_w}. The following describes the proof of correctness of Algorithm \ref{alg:maximum_value_given_w}.\par
	
	\RestyleAlgo{boxruled}
	\begin{algorithm}[!t]
		\caption{Algorithm for Problem \ref{prob:weight_sub_max_knapsack}} \label{alg:maximum_value_given_w}
		\textbf{Input:} $w>0$, $\eps>0$ and a set $V{'} = [n]$ of items, with item $i \in V{'}$ having value $v_{i}{'}$ and weight $w_{i}{'}$\\
		\textbf{Output:}  A subset $S \subseteq V{'}$ having total value at least optimal weight of Problem \ref{prob:weight_sub_max_knapsack}, such that $(1-4\eps)w \leq \sum_{i \in S} w_{i}{'} \leq (1+4\eps)w$.\\
		\begin{enumerate}
			\item 	Remove all items from $V{'}$ which have weight more than $(1+\eps)w$.
			\item \label{rounding_sub_weight_max_knapsack} For each item $i \in V{'}$,
			define \[ w_{i}{''} := \left \lceil \frac{n {w}_{i}{'}}{\eps w}  \right \rceil . \]
			\item Let $\mathcal{H}(i,w{''})$ for $ w{''} \in \left[ \left \lceil \frac{(1+2\eps)n}{\eps} \right \rceil \right]\cup \{0\}$ and $i \in [n]$, be the dynamic programming table \\ constructed in the following way.
		
			\begin{enumerate}
			\item\label{lower_bound_weight_H_maxknapsack_H1} $\mathcal{H}(1,{w{''}}_{1})={v}_{1}{'}$. $\mathcal{H}(1,0)=0$. $\mathcal{H}(1,{w{''}})=-\infty$, $\forall w{''} \in \left[ \left \lceil \frac{(1+2\eps)n}{\eps} \right \rceil \right]\setminus \{ {w}_{1}{''}\}$.  
			\item \label{lower_bound_weight_H_maxknapsack_H}
			 $\forall i\in[n]\setminus\{1\}$, $\forall w{''} \in \left[ \left \lceil \frac{(1+2\eps)n}{\eps} \right \rceil \right] $,\\ If $w{''}< w_{i}^{''}$, then $\mathcal{H}(i,w{''})=\mathcal{H}(i-1,w{''})$. \\
			Else,
			\begin{equation*}
			\begin{split}
			\mathcal{H}(i,w{''}) := \max\{&\mathcal{H}(i-1,w{''}),\mathcal{H}(i-1,w{''}-{w}_{i}{''})+v_{i}{'} \; \mid \; w_{i}{''}\leq w{''}  \}. 
			\end{split}
			\end{equation*}		
			\end{enumerate}
			
			\item \label{bound_weight_output_step_subalg_max_knapsack} Output the subset $S$ as follows,
			\[ \max \left \{ \mathcal{H}(n,w{''}) \mid \left \lceil \frac{(1-2\eps)n}{\eps} \right \rceil \leq w{''} \leq \left \lceil \frac{(1+2\eps)n}{\eps} \right \rceil  \right\} . \]
			
		\end{enumerate}
		
	\end{algorithm}

	\paragraph{Proof of Correctness.}
	We claim that the entry $\mathcal{H}(i,w{''})$,$\forall i\in [n]$, $\forall w{''} \in \left[ \left \lceil \frac{(1+2\eps)n}{\eps} \right \rceil \right]\cup \{0\}$, indicates the value of the  maximum value subset from the first $i$ items of $V{'}$ having the total rounded weight equal to $w{''}$. Let $S{'}$ denote the subset satisfying above property for the entry $\mathcal{H}(i,w{''})$,$\forall i\in [n]\setminus\{1\}$, $\forall w{''} \in \left[ \left \lceil \frac{(1+2\eps)n}{\eps} \right \rceil \right]\cup \{0\}$. 
		\begin{itemize}
	    \item If $i\in S{'}$, then $\mathcal{H}(i,w{''})$ is the sum of the value of $i$ ($v_{i}{'}$), and the maximum value subset from first $i-1$ items having the total rounded weight $w{''}-w_{i}{''}$ ($\mathcal{H}(i-1,w{''}-w_{i}{''})$).
	 \item If $i \notin S{'}$, then $\mathcal{H}(i,w{''})$ is equal to the value of maximum value subset from first $i-1$ items having the total rounded weight $w{''}$, i.e. $\mathcal{H}(i,w{''})=\mathcal{H}(i-1,w{''})$.
	 \end{itemize}
	The recursion in Step \ref{lower_bound_weight_H_maxknapsack_H} captures both of above possibilities. Step \ref{lower_bound_weight_H_maxknapsack_H1} does necessary initialization for the recursion in Step \ref{lower_bound_weight_H_maxknapsack_H}.\par
\begin{subequations}
	 \begin{align}
	    \left \lceil \frac{(1-2\eps)n}{\eps} \right \rceil  \leq \sum_{i \in S} {w}_{i}{''} & \leq \left \lceil \frac{(1+2\eps)n}{\eps} \right \rceil \\
	\Rightarrow \frac{(1-2\eps)n}{\eps}  \leq \sum_{i \in S} {w}_{i}{''} & \leq  \frac{(1+2\eps)n}{\eps}+1. \label{S_range_weight}
	 \end{align}
	 \end{subequations}
	 Using the definition of ${w}_{i}{''}$ of Step \ref{rounding_sub_weight_max_knapsack}, we obtain from the left inequality of \eqref{S_range_weight},
	  \begin{align*}
         (1-2\eps)w-\frac{\eps w}{n} & \leq \sum_{i \in S} {w}_{i}{'}\\ \Rightarrow (1-3\eps)w  & \leq \sum_{i \in S} {w}_{i}{'} .
    \end{align*}
    	 From right inequality of \eqref{S_range_weight},
    \begin{align*}
        \sum_{i \in S} {w}_{i}{''} &\leq \frac{(1+2\eps)n}{\eps}+n \\
        \Rightarrow \sum_{i \in S} {w}_{i}{'} &\leq (1+3\eps)w. &\textrm{(Definition of ${w}_{i}{''}$)}
    \end{align*}

	Now, we will prove that $\left \lceil \frac{(1-2\eps)n}{\eps} \right \rceil  \leq \sum_{i \in O} {w}_{i}{''} \leq \left \lceil \frac{(1+2\eps)n}{\eps} \right \rceil$. Since $S$ is the maximum value subset in the previous range, the above claim will imply that the total value of $S$ is at least the total value of $O$. We know that $(1-\eps)w  \leq \sum_{i \in O} {w}_{i}{'}  \leq  (1+\eps) w$. Using the definition of ${w}_{i}{''}$ in Step \ref{rounding_sub_weight_max_knapsack}, we get 
	\begin{align*}
	\frac{n(1-\eps)}{\eps}  & \leq \sum_{i \in O} {w}_{i}{''}  \leq \frac{(1+\eps)n}{\eps}+n \\
	\frac{n(1-2\eps)}{\eps} & \leq \sum_{i \in O} {w}_{i}{''}  \leq \frac{(1+2\eps)n}{\eps}\\
	\left\lceil\frac{n(1-2\eps)}{\eps}\right\rceil & \leq \sum_{i \in O} {w}_{i}{''}  \leq \left\lceil\frac{(1+2\eps)n}{\eps}\right\rceil.
	\end{align*}
	The last inequality is true by the fact that ${w}_{i}{''}$ is an integer, $\forall i \in V{'}$.
	
		Note that the algorithm returns $-\infty$ if it does not find a subset in Step \ref{bound_weight_output_step_subalg_max_knapsack}. This is because of the initialization in Step \ref{lower_bound_weight_H_maxknapsack_H1}.
	\paragraph{Running time analysis.} The size of the table $\mathcal{H}$ is $O\left(\frac{n^2}{\eps}\right)$, and to fill each entry in the table $\mathcal{H}$, we require $O(1)$ time. So, the total running time of the algorithm is $O\left(\frac{n^2}{\eps}\right)$.

\end{proof}

\begin{proof}[Proof of \Cref{the:weight_max_knapsack}]
The algorithm for theorem is described in Algorithm \ref{alg:fairness_based_lower_weight}.
The algorithm creates bundles of items from $V_i$, $\forall i \in [\NoOfCategories]$, such that the total weight of each bundle is in different ranges using Theorem \ref{the:sub_weight_max_knapsack} (Step \ref{Y_lower_max}). It stores these bundles in the table $\mathcal{Y}$. After that the algorithm combines bundles from all categories to obtain the final solution using the dynamic programming table $\mathcal{Z}$ (Step \ref{Z_max}). The total weight of each bundle is represented by some power of $\left(1+\eps{'}\right)$ in the tables $\mathcal{Y}$ and $\mathcal{Z}$. The algorithm might over calculate at most $\left(1+\eps{'} \right)$ fraction of total weight in table $\mathcal{Y}$ (Step \ref{Y_lower_max}). The total fraction of weight over calculated after combining bundles from all categories in Step \ref{Z_i_max} is at most $\left(1+\eps{'} \right)^{O(\log_{2}\NoOfCategories)}$. This is at most $\left(1+\eps\right)$ because of the choice of $\eps{'}$ in Step \ref{eps'_weight_max}. We describe the formal proof of the algorithm below.\par
	
	\RestyleAlgo{boxruled}
	\begin{algorithm}[!t]
		\caption{Algorithm for the $BW^{\max}$ (\Cref{prob:weight_max_knapsack})} \label{alg:fairness_based_lower_weight}
		\textbf{Input:} The sets $V_{i}$ of items, $\forall i\in [\NoOfCategories]$. $0\leq l_{i}^{w}\leq u_{i}^{w}$ for $i\in[\NoOfCategories]$, the knapsack capacity $B $ and $\eps>0$. \\
		\textbf{Output:}  $S$ having the total value at least the optimal value of $BW^{\max}$ (\Cref{prob:weight_max_knapsack}), such that $(1-\eps)l_{i}^{w} \leq \sum_{r \in S \cap V_i} w_{r}^{(i)} \leq (1+\eps)u_{i}^{w}$, $\forall i \in [\NoOfCategories]$, and the total weight of $S$ is at most $(1+\eps)B$.\\
		\begin{enumerate}
			
			\item \label{eps'_weight_max}	Let $\eps{'}=\left(1+\frac{3\eps}{8}\right)^{\frac{1}{\log_{2}\NoOfCategories+2}}-1$. 
			\item \label{Y_lower_max} Let $\mathcal{Y}(i,j), \ \forall i \in [\NoOfCategories]$, $\forall j \in \left[ \left \lceil \log_{1+\frac{\eps}{8}} \left(\frac{B}{w_{\min}}\right) \right \rceil\right]\cup \{0\}$ be the table, where the entry $\mathcal{Y}(i,j)$ \\indicates the value of a subset of $V_i$ that is obtained by \Cref{the:sub_weight_max_knapsack} by  setting $V_i$ as $V{'}$, $w_{\min}\left(1+\frac{\eps}{8}\right)^j$ as $w$ and $\frac{\eps}{6}$ as $\eps$ in \Cref{the:sub_weight_max_knapsack}.
			
			\item \label{Z_max}Let $\mathcal{Z}(i,j)$, $\forall i\in[2\NoOfCategories-1]$ and $\forall j \in \left[ \left \lceil \log_{1+\eps{'}} \left(\frac{B}{w_{\min}}\right) \right \rceil+\log_2 \NoOfCategories \right]\cup \{0\}$, be the DP table \\constructed as follows.
			
			\begin{enumerate}
				
				\item \label{Z_1_max} $\forall i\in[\NoOfCategories]$				\begin{equation*}
				\begin{split}
				\mathcal{Z}(m-1+i,j):=\max&\left\{\mathcal{Y}(i,j{''}) \mid j{''} \in \left[ \left \lceil \log_{1+\frac{\eps}{8}} \left(\frac{B}{w_{\min}}\right) \right \rceil   \right]\cup \{0\}\right.\\
				&\left. \& \; \left(1+\frac{\eps}{8}\right)^{j{''}+1}w_{\min} \geq l_{i}^{w}\; \& \; \left(1+\frac{\eps}{8}\right)^{j{''}-1}w_{\min} \leq u_{i}^{w} \right. \\ 
				& \left. \left(1+\frac{\eps}{8} \right)^{j{''}} \leq \left(1+\eps{'}\right)^{j}\right\} .
				\end{split}
				\end{equation*}
				If the set satisfying above condition is empty, then set $\mathcal{Z}(1,j)$ to $-\infty$.

				\item \label{Z_i_max}  $\forall i\in[\NoOfCategories-1]$ and any $j{'},j{''} \in \left[ \left \lceil \log_{1+\eps{'}} \left(\frac{B}{w_{\min}}\right) \right \rceil + \log_2 \NoOfCategories \right]\cup \{0\}$, we have

				\begin{equation*}
				\begin{split}
				\mathcal{Z}(i,j):= \max & \left\{\mathcal{Z}(2i,j{'})+\mathcal{Y}(2i+1,j{''}) \mid\;   \left(1+\eps{'}\right)^{j} \geq \left(1+\eps{'}\right)^{j{'}}+\left(1+\eps{'}\right)^{j{''}} \right \} .
				\end{split}
				\end{equation*}	
				If the set satisfying above condition is empty, then set $\mathcal{Z}(i,j)$ to $-\infty$.
			\end{enumerate}

			\item \label{lower_weight_selection_max}
			Output the subset $S$ as follows,
			\[ \max \left\{ \mathcal{Z}(\NoOfCategories,j) \mid j\in  \left[ \left \lceil \log_{1+\eps{'}} \left(\frac{B}{w_{\min}}\right) \right \rceil+\log_2 \NoOfCategories \right]\cup \{0\}  \right\}. \]
		\end{enumerate}
		
	\end{algorithm}

\paragraph{Properties of $\mathcal{Y}$. }We claim that the entry  $\mathcal{Y}(i,j)$,  $\forall i\in[\NoOfCategories]$, $\forall j \in \left[ \left \lceil \log_{1+\frac{\eps}{8}} \left(\frac{B}{w_{\min}}\right) \right \rceil \right]\cup \{0\}$, indicates the value of a subset of $V_i$ that satisfies the two properties listed below. The entry of $\mathcal{Y}$ also corresponds to respective subset. The entry of $\mathcal{Y}$ could be $-\infty$, which indicates no subset. We use the notation $\mathcal{Y}(i,j)$ to indicate both the entry and the subset.

\begin{enumerate}
	\item \label{Y_lower_weight_max_1} If the entry $\mathcal{Y}(i,j)$ is finite, then the total weight of the corresponding subset is in between $\left(1-\frac{\eps}{2}\right)$ $\left(1+\frac{\eps}{8}\right)^j w_{\min}$ and $\left(1+\frac{\eps}{2}\right)\left(1+\frac{\eps}{8}\right)^{j}w_{\min}$.
	\item \label{Y_lower_weight_max_2} The total value of the subset corresponding to $\mathcal{Y}(i,j)$ is at least the total value of any subset of $V_i$ having the total weight in between $\left(1-\frac{\eps}{6}\right)\left(1+\frac{\eps}{8}\right)^j w_{\min}$ and $\left(1+\frac{\eps}{6}\right)\left(1+\frac{\eps}{8}\right)^{j}w_{\min}$. 
\end{enumerate}

The table $\mathcal{Y}$ is created in step \ref{Y_lower_max} of the Algorithm \ref{alg:fairness_based_lower_weight}. This step uses \Cref{the:sub_weight_max_knapsack} for creation of $\mathcal{Y}$. We get both the properties of $\mathcal{Y}$ because of the guarantee of \Cref{the:sub_weight_max_knapsack}.

Let $\mathcal{T}$ be a perfect binary tree with $\NoOfCategories$ leaf nodes. For simplicity, assume that $\NoOfCategories$ is power of $2$. Although, we can prove the same result by slight modification of the proof when $\NoOfCategories$ is not power of $2$. The total number of nodes in $\mathcal{T}$ will be $2\NoOfCategories-1$.  Each node in $\mathcal{T}$ could be represented by an index number from $1$ to $2\NoOfCategories-1$ with root at index $1$.
The node at an index $i$ has an left child at an index $2i$ and right child at an index $2i+1$, $\forall i\in[\NoOfCategories-1]$. Let the leaf node at an index $(\NoOfCategories-1)+i$ represent the category $i$, $\forall i\in [\NoOfCategories]$. Let $\mathcal{T}(i)$ denote the set of categories represented by the leaves of sub tree rooted at $i$. Specifically, $\mathcal{T}(\NoOfCategories-1+i)=\{i\}$, $\forall i\in[m]$. Also, $\mathcal{T}(i)=\mathcal{T}(2i)\cup\mathcal{T}(2i+1)$, $\forall i\in [\NoOfCategories-1]$.

\paragraph{Properties of $\mathcal{Z}$. }We claim that the entry $\mathcal{Z}(i,j)$,  $\forall i\in[2\NoOfCategories-1]$, $j \in \left[ \left \lceil \log_{1+\eps{'}} \left(\frac{B}{w_{\min}}\right) \right \rceil +\log_2 \NoOfCategories \right]\cup \{0\}$, indicates the value of the subset of $\cup_{k\in \mathcal{T}(i) } V_{k}$ that satisfies the following three properties. The entry of $\mathcal{Z}$ also corresponds to the  respective subset. The entry of $\mathcal{Z}$ could be $-\infty$, which indicates no subset. We use the notation $\mathcal{Z}(i,j)$ to indicate both the entry and the subset.
\begin{enumerate}
	\item \label{Z_lower_weight_max_first} If the entry $\mathcal{Z}(i,j)$ is finite, then $\sum_{k\in \mathcal{T}(i)} \sum_{r \in \mathcal{Z}(i,j)\cap V_{k}} w^{(k)}_{r} \leq \left(1+\frac{\eps}{2}\right)\left(1+\eps{'}\right)^{j}w_{\min}$.
	\item \label{Z_lower_weight_max_second} If the entry $\mathcal{Z}(i,j)$ is finite, then $\left(1-\frac{\eps}{2}\right)\left(1+\frac{\eps}{8}\right)^{-1}l_{k}^{w} \leq \sum_{r \in V_{k} \cap \mathcal{Z}(i,j)} w_{r}^{(k)} \leq \left(1+\frac{\eps}{2}\right)\left(1+\frac{\eps}{8}\right)u_{k}^{w}$, $\forall k \in \mathcal{T}(i)$.
	\item \label{Z_lower_weight_max_third} Let $i$ be a node of $\mathcal{T}$, $\forall i \in [2\NoOfCategories-1]$, having the distance $t$ from leaves, $ t\in\left[\log_{2}\NoOfCategories\right]\cup \{0\}$.  For all $O{'} \subseteq \cup_{k\in \mathcal{T}(i)} V_{k}$ having the total weight at most $\left(1+\eps{'}\right)^{j-t}w_{\min}$, and $ l_{k}^{w} \leq \sum_{r \in  O{'}\cap V_{k}}w_{r}^{(k)}\leq u_{k}^{w}$, $\forall k \in \mathcal{T}(i)$, the total value of the subset  $\mathcal{Z}(i,j)$ is at least the total value of $O{'}$.
\end{enumerate}

Steps \ref{Z_1_max} and \ref{Z_i_max} of the algorithm chooses the subset $\mathcal{Y}(i,j{''})$ that satisfies the inequalities $\left(1+\frac{\eps}{8}\right)^{j{''}+1}$ $w_{\min}\geq l_{i}^{w}$ and $\left(1+\frac{\eps}{8}\right)^{j{''}-1}w_{\min}\leq u_{i}^{w}$. By Property \ref{Y_lower_weight_max_1} of $\mathcal{Y}$, the total weight of $\mathcal{Y}(i,j{''})$ is in between $\left(1-\frac{\eps}{2}\right)\left(1+\frac{\eps}{8}\right)^{-1}l_{i}^{w}$ and $\left(1+\frac{\eps}{2}\right)\left(1+\frac{\eps}{8}\right)u_{i}^{w}$. This proves that the subset corresponding to finite entry $\mathcal{Z}(i,j)$ satisfies Property \ref{Z_lower_weight_max_second}. The subset corresponding to finite entry $\mathcal{Z}(\NoOfCategories-1+i,j)$, $\forall i \in [\NoOfCategories]$, $\forall j \in \left[ \left \lceil \log_{1+\eps{'}} \left(\frac{B}{w_{\min}}\right) \right \rceil +\log_2 \NoOfCategories \right]\cup \{0\} $, will satisfy Property \ref{Z_lower_weight_max_first}. This is true because of Property \ref{Y_lower_weight_max_1} of $\mathcal{Y}$ and the condition $\left( 1+\frac{\eps}{8} \right)^{j''} \leq \left( 1+\eps{'} \right)^{j}$ in Step \ref{Z_1_max} for selecting subset $\mathcal{Z}(i,j{''})$. Because of the condition $\left(1+\eps{'}\right)^{j}\geq (1+\eps{'})^{j{'}}+ (1+\eps{'})^{j{''}}$ in Step \ref{Z_i_max} of the algorithm and Property $\ref{Z_lower_weight_max_first}$ of $\mathcal{Z}$, the subset corresponding to finite entry $\mathcal{Z}(i,j)$, $\forall i\in [\NoOfCategories-1]$, $\forall j \in \left[ \left \lceil \log_{1+\eps{'}} \left(\frac{B}{w_{\min}}\right) \right \rceil +\log_2 \NoOfCategories \right]\cup \{0\}$, will satisfy Property \ref{Z_lower_weight_max_first}.\par

We prove that the nodes in $\mathcal{T}$ will satisfy Property \ref{Z_lower_weight_max_third} by induction. In the base case, we prove that Property \ref{Z_lower_weight_max_third} is satisfied for all leaves. Let $O{''}$ be any subset of $V_i$ that satisfies the fairness bounds such that the total weight of $O{''}$ is at most $w_{\min}\left(1+\eps{'} \right)^{j}$. Let $j{''}\in \left[ \left \lceil \log_{1+\frac{\eps}{8}} \left(\frac{B}{w_{\min}}\right) \right \rceil \right]\cup \{0\}$ such that the total weight of $O{''}$ is in between $w_{\min}\left(1+\frac{\eps}{8}\right)^{j{''}-1}$ and $w_{\min}\left(1+\frac{\eps}{8}\right)^{j{''}}$. By Property \ref{Y_lower_weight_max_2} of $\mathcal{Y}$, the total value of $\mathcal{Y}(i,j{''})$ is at least $O{''}$. Since $O{''}$ satisfies the fairness bounds, the conditions $w_{\min}\left(1+\frac{\eps}{8}\right)^{j{''}+1}\geq l_{i}^{w}$ and $w_{\min}\left(1+\frac{\eps}{8}\right)^{j{''}-1}\leq u_{i}^{w}$ in Step \ref{Z_1_max} are also satisfied. So, the subset $\mathcal{Y}(i,j{''})$ is feasible for $\mathcal{Z}(\NoOfCategories-1+i,j)$ in Step \ref{Z_1_max}. So, the Property \ref{Z_lower_weight_max_third} is satisfied by $\mathcal{Z}(\NoOfCategories-1+i,j)$ for all $i\in[\NoOfCategories]$.\par
For any $t\in\left[\log_{2}\NoOfCategories-1\right]\cup \{0\}$, assume the hypothesis that all the nodes having distance $t$ from leaves satisfy Property \ref{X_lower_value_max_third}. We will prove by induction that for any node $i$ with distance $t+1$ from leaves, Property \ref{X_lower_value_max_third} is satisfied. Let $O{'}\subseteq \cup_{k\in \mathcal{T}(i)} V_k$ that satisfies the fairness conditions for all categories in $\mathcal{T}(i)$.
Let $j{^*} \in  \left[ \left \lceil \log_{1+\eps{'}} \left(\frac{B}{w_{\min}}\right) \right \rceil +\log_2 \NoOfCategories \right]\cup \{0\}$ satisfies the following inequality

\begin{equation}\label{eq:bound_weight_max_knapsack_1}
\left(1+\eps{'}\right)^{j{^*-1}}w_{\min} \leq \sum_{k\in \mathcal{T}(2i)}\sum_{r \in O{'}\cap V_k}w^{(k)}_{r} \leq\left(1+\eps{'}\right)^{j{^*}}w_{\min}.    
\end{equation}
If $j'=j{^*}-t$,
the induction hypothesis implies the following :
\begin{equation}\label{eq:value_upper_bound_left_child}
\sum_{k\in \mathcal{T}(2i)} \sum_{r \in \mathcal{X}(2i,j{'})\cap V_{k}} v^{(k)}_{r} \geq \sum_{k\in \mathcal{T}(2i)} \sum_{r \in O{'} \cap V_{k}} w^{(k)}_{r}.
\end{equation}
Similarly, let $j{^{**}} \in  \left[ \left \lceil \log_{1+\eps{'}} \left(\frac{B}{w_{\min}}\right) \right \rceil +\log_2 \NoOfCategories \right]\cup \{0\}$ that satisfies the following inequality 

\begin{equation}\label{eq:bound_weight_max_knapsack_2}
\left(1+\eps{'}\right)^{j{^{**}}-1}w_{\min} \leq \sum_{k\in \mathcal{T}(2i+1)}\sum_{r \in O{'}\cap V_k}w^{(k)}_{r} \leq\left(1+\eps{'}\right)^{j{^{**}}}w_{\min}.    
\end{equation}
If $j''=j{^{**}}-t$, the induction hypothesis implies the following \begin{equation}\label{eq:value_upper_bound_right_child}
\sum_{k\in \mathcal{T}(2i+1)} \sum_{r \in \mathcal{X}(2i+1,j{''})\cap V_{k}} v^{(k)}_{r} \geq \sum_{k\in \mathcal{T}(2i+1)} \sum_{r \in O{'} \cap V_{k}}v^{(k)}_{r}.    
\end{equation}

We claim that the inequality $\left(1+\eps{'}\right)^{j}\geq \left(1+\eps{'}\right)^{j{'}}+ \left(1+\eps{'}\right)^{j{''}} $ is satisfied in Step \ref{Z_i_max} for any $j\in  \left[ \left \lceil \log_{1+\eps{'}} \left(\frac{B}{w_{\min}}\right) \right \rceil +\log_2 \NoOfCategories \right]\cup \{0\}$, such that the total weight of $O{'}$ is at most $w_{\min}(1+\eps{'})^{j-t-1}$. This is true because the minimum weight of $O{'}$ is at least $w_{\min}\left(\left(1+\eps{'}\right)^{j{^*}-1}+\left(1+\eps{'}\right)^{j{^{**}}-1} \right)$ (by Inequality \ref{eq:bound_weight_max_knapsack_1} and Inequality \ref{eq:bound_weight_max_knapsack_2}), which is a most $w_{\min}\left(1+\eps{'}\right)^{j-t-1}$. So, the pair of subsets $\mathcal{Z}(2i,j{'})$ and $\mathcal{Z}(2i+1,j{''})$ is feasible in Step \ref{X_i_max} for all such $j$. By Equation \ref{eq:value_upper_bound_left_child} and Equation \ref{eq:value_upper_bound_right_child}, Property \ref{X_lower_value_max_third} of $\mathcal{Z}$ is satisfied for $\mathcal{Z}(i,j)$ and $O{'}$.\par	

By property \ref{Z_lower_weight_max_third} of $\mathcal{Z}$, there exists an subset $\mathcal{Z}(1,j)$, for some $j\in  \left[ \left \lceil \log_{1+\eps{'}} \left(\frac{B}{w_{\min}}\right) \right \rceil+\log_2 \NoOfCategories \right]\cup \{ 0\}$, that satisfies the condition in Property \ref{Z_lower_weight_max_third} for some optimal solution. This proves that the total value of $S$ (Step \ref{lower_weight_selection_max}) is at least the optimal value. By property \ref{Z_lower_weight_max_first} of $\mathcal{Z}$, the total weight of $S$ (Step \ref{lower_weight_selection_max}) is at most $\left(1+\frac{\eps}{2} \right)(1+\eps{'})^{\log_2 \NoOfCategories+1}B$. By Step \ref{eps'_weight_max} of the algorithm, this number is $(1+\frac{\eps}{2} )(1+\frac{3\eps}{8})B < (1+\eps)B$. \par
	
If the entry $\mathcal{Z}(1,j)$ is finite, Property \ref{X_lower_value_max_second} of $\mathcal{Z}$ implies that the total weight of items from $V_i \cap \mathcal{Z}(1,j)$ is in between $\left(1-\frac{\eps}{2} \right)\left(1+\frac{\eps}{8} \right)^{-1}l_{i}^{w}$ and $\left(1+\frac{\eps}{2} \right)\left(1+\frac{\eps}{8} \right)u_{i}^{w}$ for all $i\in [\NoOfCategories]$. The total weight of items from $V_i \cap \mathcal{Z}(1,j)$ is at least $\left(1-\frac{\eps}{2} \right)\left(1+\frac{\eps}{8} \right)^{-1}l_{i}^{w}> \left(1-\frac{\eps}{2} \right)\left(1-\frac{\eps}{8} \right) l_{i}^{w}> (1-\eps)l_{i}^{w}$. The total weight of items from $V_i \cap \mathcal{Z}(1,j)$ is at most $\left(1+\frac{\eps}{2} \right)\left(1+\frac{\eps}{8} \right)u_{i}^{w} < \left(1+\frac{6}{8}\eps \right)u_{i}^{w}<\left(1+\eps \right)u_{i}^{w}$. \par

	\paragraph{Running time Analysis:} The size of the table $\mathcal{Y}$ is $O \left(\NoOfCategories \log_{(1+\eps)} \left(\frac{B}{w_{\min}}\right)\right)$, and we require $O(\frac{n^2}{\eps})$ time to fill each entry. So, the total time required to build the table $\mathcal{Y}$ is $O \left( \frac{n^2 \NoOfCategories \log_{(1+\eps)} \left(\frac{B}{w_{\min}}\right)}{\eps}\right)$. The total time required to build the table $\mathcal{Z}$ from the table $\mathcal{Y}$ is $O\left(\NoOfCategories \log^3 \NoOfCategories \log_{(1+\eps)}^{3} \left(\frac{B}{w_{\min}}\right)\right)$. So, the total running time is $O \left( \frac{n^2 \NoOfCategories \log^3 \NoOfCategories \log_{(1+\eps)}^{3} \left(\frac{B}{w_{\min}}\right)}{\eps}\right)$.
\end{proof}

%% file: algorithm_min_knapsack.tex
\section{Fair min-knapsack}\label{Fair min-knapsack}
The classical min-knapsack problem is to find the packing of minimum weight having the total value at least the given lower bound.\par

We will consider the same fairness notions in the min-knapsack as we have considered for the max-knapsack case. We give algorithms for these notions of fairness in the min-knapsack case in this section. 
\subsection{Fairness based on number of items}

\begin{problem}[$BN^{min}$]\label{prob:number_min}
Given a set of items, each belonging to  one of $\NoOfCategories$ categories, and numbers $l_{i}^{n}$ and $u_{i}^{n}$ for category $i$, $\forall i \in [\NoOfCategories]$, the problem is to find a subset that minimizes the total weight, such that the number of items from category $i$ is between $l_{i}^{n}$ and $u_{i}^{n}$, $\forall i\in[\NoOfCategories]$, and the total value of the subset is at least the given value lower bound $L$. 
\end{problem}
We prove the following theorem.
\begin{theorem}\label{the:number_min}
For any $\eps>0$, there exists a $(1+\eps)$-approximation algorithm for $BN^{min}$ (Problem \ref{prob:number_min}) with running time $O\left(\frac{(1+\eps)^2 n^3\NoOfCategories^3}{{\eps}^2} \log_{(1+\eps)} \left( \frac{\NoOfItmes w_{max}}{w_{min}}\right)\right)$,  $w_{max}:=\max\left\{ w_{j}^{(i)} \mid i\in[\NoOfCategories] \; \& \; j \in V_i\; \& \;  w_{j}^{(i)}>0 \right\}$ and $w_{min}:=\min\left\{w_{j}^{(i)} \mid i\in[\NoOfCategories] \; \& \; j \in V_i \; \& \; w_{j}^{(i)}>0 \right\}$. 
\end{theorem}
\begin{proof}
The algorithm approximately guesses the optimal weight in the beginning (Step \ref{guess_item_number_min}). According to the guess, it rounds the weights of all items so that the rounded weights lie in a small range (Step \ref{rounding_number_min}). Then it creates bundles of items from $V_i$, $\forall i \in [\NoOfCategories]$, having different rounded weights and cardinality (Step \ref{A_number_min}). It uses the dynamic programming table $\mathcal{A}$ for this. After that the algorithm combines the bundles from all categories to obtain the final solution using the dynamic programming table $\mathcal{B}$ (Step \ref{B_number_min}). The algorithm is described in Algorithm \ref{alg:min_fairness_based_number}. The following describes the proof of correctness of Algorithm \ref{alg:min_fairness_based_number}.\par	
	
	\RestyleAlgo{boxruled}
	\begin{algorithm}
		\caption{Algorithm for $BN^{min}$ (Problem \ref{prob:number_min})} \label{alg:min_fairness_based_number}
		\textbf{Input:} The sets $V_{i}$ of items, $\forall i\in [\NoOfCategories]$. $0 \leq l_{i}^{n}\leq u_{i}^{n}$ for $i\in[\NoOfCategories]$, and $\eps>0$, and value lower bound $L$ \\
		\textbf{Output:}  $S$ having the total weight at most $1+\eps$ times the weight of the optimal solution of Problem \ref{prob:number_min}, such that $l_{i}^{n} \leq \mid S \cap V_i\mid \leq u_{i}^{n}$, $\forall i \in [\NoOfCategories]$ and the total value of of $S$ is at least the given bound $L$.\\
		\begin{enumerate}
			\item \label{guess_item_number_min} If $\OPT$ is the weight of the optimal solution, then we can guess $b_{opt}$ such that $b_{opt} \leq \OPT \leq (1+\eps)b_{opt}$,
			\item \label{rounding_number_min}$\forall i\in [\NoOfCategories]$ and $\forall j \in  V_i$, let ${w}_{j}^{(i)}{'}:= \left \lceil \frac{w_{j}^{(i)}\NoOfItmes}{\eps b_{opt}} \right\rceil.$
			
			\item \label{A_number_min}Let $\mathcal{A}(i,j,w,t)$, $\forall i \in [\NoOfCategories]$, $\forall j \in   V_i $, $\forall t \in V_i\cup \{0\}$, $\forall w \in \left[ \left \lceil \frac{\NoOfItmes (1+2\eps)}{\eps}\right \rceil \right] \cup \{0\}$, be the dynamic programming table constructed in the following way,
			
			\begin{enumerate}
			\label{Ai_number_max} 
			
			\item \label{A1_number_min}$\forall i\in [\NoOfCategories]$,\\ $\mathcal{A}(i,1,{w}^{(i)}_{1}{'},1)=v^{(i)}_{1}$, $\mathcal{A}(i,1,0,0)=0$. $\mathcal{A}(i,1,w,.)=-\infty$, $\forall w \in \left[ \left \lceil \frac{\NoOfItmes (1+2\eps)}{\eps}\right \rceil \right] \setminus \{w_{1}^{(i)'}\}$.
			
			\item \label{Ai_number_min} $\forall i\in [\NoOfCategories]$ and $\forall j \in V_i\setminus\{1\}$, $\forall t\in V_i \cup \{0\}$, $\forall w\in \left[ \left \lceil \frac{\NoOfItmes (1+2\eps)}{\eps}\right \rceil \right] \cup \{0\}$,
			
				If $w<w_{j}^{(i)}$ or $t=0$, then $\mathcal{A}(i,j,w,t)= \mathcal{A}(i,j-1,w,t) $.\\
				Else ,
			\begin{equation*}
			\begin{split}
			\mathcal{A}(i,j,w,t):=\max&\left \{ \mathcal{A}(i,j-1,w-{w}_{j}^{(i)}{'},t-1)+v_{j}^{(i)}, \mathcal{A}(i,j-1,w,t) \mid \right. \\ 
			& \left. 0 \leq {w}_{j}^{(i)}{'} \leq w \;
			\&
			\;  {w}_{j}^{(i)}{'}\in \left[\left \lceil \frac{\NoOfItmes(1+2\eps)}{\eps}\right \rceil\right] \cup \{0\}\right\}. 
			\end{split}
			\end{equation*}

			\end{enumerate}

			\item \label{B_number_min}  Let $\mathcal{B}(i,w)$, $\forall i\in[\NoOfCategories]$, $\forall w \in \left[\left \lceil\frac{\NoOfItmes (1+2\eps)}{\eps}\right \rceil\right]\cup \{0\}$ be another dynamic programming table.
			
			\begin{enumerate}
			\item \label{B1_number_min}
			
			\[\mathcal{B}(1,w):=\max\{\mathcal{A}(1,\mid V_1\mid,w,t)\; |\;l_{1}^{n}\leq t\leq u_{1}^{n}\}.\]
			
			\item \label{BI_number_min}For $i \in [\NoOfCategories]\setminus\{1\}$,
			
			\begin{equation*}
			\begin{split}
			\mathcal{B}(i,w):= \max & \left\{ \mathcal{B}(i-1,w_r)+\mathcal{A}(i,\mid V_i \mid,w-w_r,t) \mid l_{i}^{n}\leq t\leq u_{i}^{n} \right. \\
			& \left. w_r \leq w \; \& \; w_r \in \left[\left \lceil\frac{\NoOfItmes(1+2\eps)}{\eps}\right \rceil\right] \cup \{0\} \right\}.
			\end{split}
			\end{equation*}

			\end{enumerate}

			\item \label{output_min_number}	Output the subset $S$ as follows,
			\[ \argmin_{w} \{ \mathcal{B}(\NoOfCategories,w) \mid \mathcal{B}(\NoOfCategories,w) \geq L \}. \]

		\end{enumerate}
		
	\end{algorithm}
	
	\paragraph{Property of $\mathcal{A}$. }We claim that $\mathcal{A}(i,j,w,t)$, $\forall i \in [\NoOfCategories]$, $\forall t\in V_i \cup \{0\}$, $j\in V_i$, $\forall w \in \left[ \left \lceil  \frac{\NoOfItmes\left(1+2\eps \right)}{\eps} \right \rceil \right] \cup \{0\}$, denotes the value of a maximum value subset of cardinality $t$ from first $j$ items of $V_i$ having total rounded weight $w$. Let $S{'}$ be the subset of $V_i$ satisfying above property for the entry $\mathcal{A}(i,j,w,t)$, $\forall i \in [\NoOfCategories]$, $\forall j \in V_i\setminus\{1\}$, $\forall t \in V_i\cup \{0\}$, $\forall w \in \left[ \left \lceil  \frac{\NoOfItmes\left(1+2\eps \right)}{\eps} \right \rceil \right] \cup \{0\}$.
	
	\begin{itemize}
	    \item If $j\in S{'}$, then $\mathcal{A}(i,j,w,t)$ is the sum of value of $j$ ($v_{j}^{(i)}$) and the value of maximum value subset from first $j-1$ items of $V_i$ having cardinality $t-1$ and the rounded weight $w-w_{j}^{(i)}{'}$ ($\mathcal{A}(i,j-1,w-w_{j}^{(i)}{'},t-1)$), which is equal to $\mathcal{A}(i,j-1,w-w_{j}^{(i)}{'},t-1)+v_{j}^{(i)}$.
	    \item If $j\notin S{'}$, then $\mathcal{A}(i,j,w,t)$ is equal to the value of maximum value subset from first $j-1$ items of $V_i$ having cardinality $t$ and the rounded weight $w$, which is equal to $\mathcal{A}(i,j-1,w,t)$.  
	\end{itemize}
	 The recursion in Step \ref{Ai_number_min} captures both of above possibilities. The Step \ref{A1_number_min} initializes base case for the recursion in Step \ref{Ai_number_min}. The entry of $\mathcal{A}$ also corresponds to respective subset. The entry of $\mathcal{A}$ could be $-\infty$, which indicates no subset. We use the notation $\mathcal{A}(i,j,w,t)$ to indicate both the entry and the subset.
	 
	\paragraph{Property of $\mathcal{B}$. }We claim that $\mathcal{B}(i,w)$, $\forall i \in [\NoOfCategories]$, $\forall w  \in \left[ \left \lceil  \frac{\NoOfItmes\left(1+2\eps \right)}{\eps} \right \rceil \right] \cup \{0\}$, denotes the value of a maximum value subset of $\cup_{j=1}^{i} V_j$ having the total rounded weight $w$, such that the fairness constraints are satisfies for all categories up to $i$. Let $S{'}$ be the subset of $\cup_{j=1}^{i} V_j$ satisfying above property for the entry $\mathcal{B}(i,w)$, $\forall i \in [\NoOfCategories]\setminus\{1\}$, $\forall w  \in \left[ \left \lceil  \frac{\NoOfItmes\left(1+2\eps \right)}{\eps} \right \rceil \right] \cup \{0\}$. If $\sum_{j\in S{'}\cap V_i}w_{j}^{(i)}{'}=w_r$, then $\mathcal{B}(i,w)$ is sum of values of maximum value subset of $\cup_{j=1}^{i-1}V_j$ having the total rounded weight $w-w_r$ that satisfies fairness condition for all categories up to $i-1$ ($\mathcal{B}(i-1,w-w_r)$), and the maximum value subset of $V_i$ having the total rounded weight $w_r$ that satisfies fairness condition for category $i$ ($\mathcal{A}(i,\mid V_i\mid,w_r,t)$ such that $l_{i}^{n}\leq t \leq u_{i}^{n}$). The recursion in Step \ref{BI_number_min} captures this for all possible $w_r$. The entry of $\mathcal{B}$ also corresponds to respective subset. The entry of $\mathcal{B}$ could be $-\infty$, which indicates no subset. We use the notation $\mathcal{B}(i,w)$ to indicate both the entry and the subset.  \par

	If $\OPT$ is the weight of the optimal solution of $BN^{min}$ (\Cref{prob:number_min}), then we can guess $b_{opt}$ such that $b_{opt} \leq OPT \leq (1+\eps)b_{opt}$, in time $O\left(\log_{1+\eps} \left( \frac{\NoOfItmes w_{max}}{w_{min}}\right)\right)$. This is true because of the inequality $w_{min} \leq OPT \leq \NoOfItmes w_{max}$. $S$ (Step \ref{output_min_number}) will be fair for all categories because of the property of $\mathcal{B}$. Also, $S$ will satisfy the knapsack value lower bound constraint because of the condition in the Step \ref{output_min_number}. \Cref{lema:number_min} proves this theorem. \par
	
	\paragraph{Running time analysis:} There are $O\left(\log_{1+\eps} \left( \frac{\NoOfItmes w_{max}}{w_{min}}\right)\right)$ possible values of $b_{opt}$ in Step \ref{guess_item_number_min}. The size of the table $\mathcal{A}$ is $O\left(\frac{(1+\eps)n\NoOfItmes^2}{\eps}\right)$, and we require $O\left(1\right)$ time to fill the each entry in it. The size of the table $\mathcal{B}$ is $O\left(\frac{(1+\eps)\NoOfCategories\NoOfItmes}{\eps}\right)$. We require $O\left(\frac{(1+\eps)n\NoOfItmes}{\eps}\right)$ time to fill each entry of the table $\mathcal{B}$. So, the total time required by the algorithm is $O\left(\frac{(1+\eps)^2 n^3\NoOfCategories^3}{\eps^2} \log_{1+\eps} \left( \frac{\NoOfItmes w_{max}}{w_{min}}\right)\right)$.
\end{proof}

\begin{theorem}\label{lema:number_min}
	If $\OPT$ is the weight corresponding to the optimal solution of $BN^{min}$ (Problem \ref{prob:number_min}) and $b_{opt}$ satisfies the inequality $b_{opt} \leq \OPT \leq (1+\eps)b_{opt}$, then the total weight of set returned in Step \ref{output_min_number} of Algorithm \ref{alg:min_fairness_based_number} is at most $(1+\eps)\OPT$.
\end{theorem}
\begin{proof}
	Let $O \subseteq \cup_{i=1}^{\NoOfCategories} V_{i}$ be the set of items in the optimal solution, and $S \subseteq \cup_{i=1}^{\NoOfCategories} V_{i}$ be the set of items selected by Step \ref{output_min_number} of Algorihtm \ref{alg:min_fairness_based_number}. The total value of $O$ and $S$ is at least $L$ (Step \ref{output_min_number}). Because of the rounding in Step \ref{rounding_number_min} and the inequality $\OPT \leq (1+\eps)b_{opt}$, the total rounded weight of $O$ is in range $\left[ \left \lceil  \frac{\NoOfItmes\left(1+2\eps \right)}{\eps} \right \rceil \right] \cup \{0\}$. Since $S$ is the subset with total value at least $L$ having minimum rounded weight in the previous range (Step \ref{output_min_number}), the total rounded weight of $S$ is at most the total rounded weight of $O$.
	\begin{equation}\label{eq:rounded_weight_S_O_min}
	\sum_{i=1}^{\NoOfCategories}  \sum_{j \in V_i \cap O} {w}_{j}^{(i)}{'} \geq \sum_{i=1}^{\NoOfCategories} \sum_{j \in V_i \cap S} {w}_{j}^{(i)}{'}.
	\end{equation}    
	
	As per the Step \ref{rounding_number_min} of Algorithm \ref{alg:min_fairness_based_number},  $\forall i\in [\NoOfCategories]$ and $j\in V_i$,
	\begin{equation}\label{eq:rounded_eq_min}
	\frac{\eps b_{opt} ({w}_{j}^{(i)}{'}-1)}{\NoOfItmes} = \frac{\eps b_{opt} {w}_{j}^{(i)}{'}}{\NoOfItmes}-\frac{\eps b_{opt}}{\NoOfItmes} \leq w_{j}^{(i)} \leq \frac{\eps b_{opt} ({w}_{j}^{(i)}{'})}{\NoOfItmes}.
	\end{equation}
	
	So we get,
	
	\begin{align*}
	\sum_{i=1}^{\NoOfCategories}  \sum_{j\in V_i \cap O} w_{j}^{(i)} & \geq \frac{\eps b_{opt}}{\NoOfItmes} \left(\sum_{i=1}^{\NoOfCategories} \sum_{j \in V_i \cap O} ({w}_{j}^{(i)}{'}-1) \right) &\textrm{(Inequality \ref{eq:rounded_eq_min})}\\
	& \geq \frac{\eps b_{opt}}{\NoOfItmes} \left(\sum_{i=1}^{\NoOfCategories} \sum_{j \in V_i \cap O} {w}_{j}^{(i)}{'} \right)- \eps b_{opt}&\textrm{$\left(\sum_{i=1}^{l} \sum_{j \in V_i \cap S} 1 \leq \NoOfItmes \right)$} \\
	& \geq \frac{\eps b_{opt}}{\NoOfItmes} \left(\sum_{i=1}^{\NoOfCategories} \sum_{j \in V_i \cap S} {w}_{j}^{(i)}{'} \right)-\eps b_{opt} &\textrm{(Inequality \ref{eq:rounded_weight_S_O_min})}\\
	&\geq \left(\sum_{i=1}^{\NoOfCategories} \sum_{j \in V_i \cap S} {w}_{j}^{(i)} \right)-\eps b_{opt} &\textrm{(Inequality \ref{eq:rounded_eq_min})}
	\end{align*}
	 If the guess of $b_{opt}$ is correct, then we have $b_{opt} \leq \OPT$. So,
	\[ \left(\sum_{i=1}^{\NoOfCategories} \sum_{j \in V_i \cap S} w_{j}^{(i)} \right) \leq (1+\eps)\OPT. \]

\end{proof}

\subsection{Fairness based on bound on value}

\begin{problem}[$BV^{min}$]\label{prob:value_min}
Given a set of items, each belonging to  one of $\NoOfCategories$ categories, and numbers $l_{i}^{v}$ and $u_{i}^{v}$ for category $i$, $\forall i\in[\NoOfCategories]$, the problem is to find a subset that minimizes the total weight, such that the total value of items from category $i$ is between $l_{i}^{v}$ and $u_{i}^{v}$, $\forall i\in[\NoOfCategories]$, and the total value of the subset is at least the given value lower bound $L$.  
\end{problem}

We prove that it is $\NP$-hard to obtain a feasible solution of an instance of $BV^{min}$ (\Cref{prob:value_min}).

\begin{theorem}\label{hardness_result_bound_value_min}
	There is no polynomial time algorithm which outputs a feasible solution of $BV^{min}$ (\Cref{prob:value_min}), assuming $\PP \neq \NP$.
\end{theorem}
\begin{proof}
	Given an instance of subset sum (Problem \ref{subsetSum2}), we can construct an instance of $BV^{min}$ (\Cref{prob:value_min}) in the following way. Let $\NoOfCategories=1$, $l_{1}^{v}=0$, $u_{1}^{v}=\frac{1}{2}$ and $L=\frac{1}{2}$. The set $V_1$ contains items which correspond to the numbers in $I$. An item corresponding to some $a\in I$ has a value and a weight equal to $a$. This proves that if we have an algorithm that outputs a feasible solution of $BV^{min}$ (\Cref{prob:value_min}), then we can solve subset sum (\Cref{subsetSum2}). But assuming $\PP\neq \NP$, this is not possible.	
\end{proof}

Theorem \ref{hardness_result_bound_value_min} implies that there does not exists polynomial time algorithm for $BV^{min}$ (\Cref{prob:value_min}). We give here an algorithm for $BV^{min}$ (\Cref{prob:value_min}) that might violates fairness constraints and knapsack constraint by a small amount.

\begin{theorem}\label{the:value_min}
		For any $\eps>0$, there exists an algorithm for $BV^{min}$ (\Cref{prob:value_min}) that outputs a set $S$ having the total weight at most the optimal weight of $BV^{min}$ (\Cref{prob:value_min}), such that $(1-\eps)l_{k}^{v} \leq \sum_{r \in S \cap V_k} v^{(k)}_r \leq (1+\eps)u_{k}^{v}$, $\forall k \in [\NoOfCategories]$, and the total value of items in $S$ is at least $\left( 1-\eps\right)L$. The running time of the algorithm is $O \left( \frac{n^2\NoOfCategories \log^3 \NoOfCategories \log_{1+\eps}^{3}\left (\frac{\NoOfItmes v_{max}}{v_{min}}\right)}{\eps}\right)$, where
	$v_{min}:=\min \left\{ v_{j}^{(i)} \mid i \in [\NoOfCategories] \; \& \; j \in  V_i \; \& \; v_{j}^{(i)}>0  \right\}$. and $v_{max}:=\max \left\{ v^{(i)}_{j} \mid i \in [\NoOfCategories] \; \& \; j \in V_i   \right\}$.
\end{theorem}
\begin{proof}

The algorithm creates bundles of items from $V_i$, $\forall i \in [\NoOfCategories]$, such that the total value of each bundle is in different ranges using Theorem \ref{the:sub_value_max_knapsack} (Step \ref{W_min}). It stores these bundles in the table $\mathcal{W}$. After that the algorithm combines bundles from all categories in divide and conquer fashion to obtain the final solution using the dynamic programming table $\mathcal{X}$ (Step \ref{X_min}). The total value of each bundle is represented by some power of $\left(1+\eps{'}\right)$ in the tables $\mathcal{W}$ and $\mathcal{X}$. So, the algorithm might lose at most $\left(1+\eps{'} \right)$ fraction of total value in one iteration of Step \ref{X_min}. The total fraction of value lost in the calculation after combining the bundles from all the categories in Step \ref{X_i_min} is at most $\left(1+\eps{'} \right)^{O(\log_{2}\NoOfCategories)}$. This is at most $\left(1+\eps\right)$ because of the choice of $\eps{'}$ in Step \ref{round_eps_value_max}. We describe the formal proof of the algorithm below.\par

		\RestyleAlgo{boxruled}
	\begin{algorithm}[!t]
		\caption{Algorithm for the $BV^{min}$ (\Cref{prob:value_min})} \label{alg:fairness_based_value_min}
		\textbf{Input:} The sets $V_{i}$ of items, $\forall i\in [\NoOfCategories]$, $0 \leq l_{i}^{v}\leq u_{i}^{v}$, $\forall i\in[\NoOfCategories]$, the value lower bound $L$ and $\eps>0$. \\
		\textbf{Output:} The subset $S$ of items, having the total weight at most the optimal weight of $BV^{min}$ (\Cref{prob:value_min}), such that $(1-\eps)l_{i}^{v} \leq \sum_{r \in S \cap V_i} \leq (1+\eps)u_{i}^{v} $, $\forall i \in [\NoOfCategories]$, and the total value of $S$ is at least $\left(1-\eps\right)L$.\\
		\begin{enumerate}
			\item \label{eps'_value_min}	Let $\eps{'}=\left(1+\frac{3\eps}{8}\right)^{\frac{1}{\log_2\NoOfCategories+1}}-1$.  Also let $W_{range}:=\left[ \left \lceil \log_{1+\frac{\eps}{8}} \left(\frac{\NoOfItmes v_{max}}{v_{min}}\right) \right \rceil \right]\cup \{0\}$  and
 			$X_{range}:=\left[ \left \lceil \log_{1+\eps{'}} \left(\frac{\NoOfItmes v_{max}}{v_{min}}\right) \right \rceil \right] \cup \{-\log_2 \NoOfCategories,-\log_2 \NoOfCategories+1,...-2,-1,0\}$. 
			
			\item \label{W_min} Let $\mathcal{W}(i,j), \forall i \in [\NoOfCategories]$, $\forall j \in W_{range}$ be the table where the entry $\mathcal{W}(i,j)$ indicates the weight of \\a subset of $V_i$ that is obtained by \Cref{the:sub_value_max_knapsack} by setting $V_i$ as $V{'}$, $v_{min}\left(1+\frac{\eps}{8}\right)^j$ as $v$ and $\frac{\eps}{6}$ as \\$\eps$ in \Cref{the:sub_value_max_knapsack}.
		\item \label{X_min}Let $\mathcal{X}(i,j)$, $\forall i\in [2\NoOfCategories-1]$, $\forall j \in X_{range}$, be the DP table constructed as follows.
			
			\begin{enumerate}
			\item \label{X_1_min}
			$\forall i \in [\NoOfCategories]$, $\forall j\in X_{range}$,
			\begin{equation*}
			\begin{split}
			\mathcal{X}(\NoOfCategories-1+i,j):=\min & \left \{\mathcal{W}(i,j{''}) \mid  \;
		   \left(1+\frac{\eps}{8}\right)^{j{''}}\geq \left(1+\eps{'}\right)^{j}  \; \& \; v_{min}\left(1+\frac{\eps}{8}\right)^{j{''}+1}  \right.\\ 
		   &\left.  \geq l_{i}^{v} \; \& \; v_{min}\left(1+\frac{\eps}{8}\right)^{j{''}-1}\leq u_{i}^{v}\; \&  \;  j{''} \in W_{range}   \right\} .
			\end{split}
			\end{equation*}
			
			 If the set satisfying above condition is empty, then set $\mathcal{X}(i,j)$ to $\infty$.
			\item \label{X_i_min} $\forall i\in[\NoOfCategories-1]$, $\forall j{'},j{''} \in X_{range}$, we have
			
			\begin{equation*}
			\begin{split}
			\mathcal{X}(i,j) := \min & \left\{ \mathcal{X}(2i,j{'})+\mathcal{X}(2i+1,j{''}) \mid \; \left(1+\eps{'}\right)^{j} \leq \left(1+\eps{'}\right)^{j{'}}+\left(1+\eps{'}\right)^{j{''}} \right\} .
			\end{split}
			\end{equation*}
			If the set satisfying above condition is empty, then set $\mathcal{X}(i,j)$ to $\infty$.
			
			\end{enumerate}

		\item \label{Output_value_min} Output the subset $S$ as follows, 
		\[ \min \left\{ \mathcal{X}(1,j) \mid v_{min}\left(1+\eps{'}\right)^{j+\log_2 \NoOfCategories} \geq L \right\} . \]	
		\end{enumerate}
		
	\end{algorithm}
	
The tables $\mathcal{W}$ (Step \ref{W_min}) and $\mathcal{X}$ (Step \ref{X_min}) are same as the tables created by the algorithm (Algorithm \ref{alg:fairness_based_value}) in \Cref{the:alg_value_max_knapsack}. As proved in \Cref{the:alg_value_max_knapsack}, the entries in these tables satisfy following properties.

\paragraph{Properties of $\mathcal{W}$.}We claim that the entry $\mathcal{W}(i,j)$,  $\forall i\in[\NoOfCategories]$, $j \in W_{range}$, indicates the weight of the subset of $V_i$ that satisfies the two properties listed below. The entry of $\mathcal{W}$ also corresponds to respective subset. The entry of $\mathcal{W}$ could be $\infty$, which indicates no subset. We use the notation $\mathcal{W}(i,j)$ to indicate both the entry and the subset. 
\begin{enumerate}
	\item \label{W_lower_value_min_1} If the entry $\mathcal{W}(i,j)$ is finite, then the total value of the corresponding subset is in between $\left(1-\frac{\eps}{2}\right)$ $\left(1+\frac{\eps}{8}\right)^j v_{min}$ and $\left(1+\frac{\eps}{2}\right)\left(1+\frac{\eps}{8}\right)^{j}v_{min}$.
	\item \label{W_lower_value_min_2} The total weight of the subset corresponding to $\mathcal{W}(i,j)$ is at most the total weight of any subset of $V_i$ having the total value in between $v_{min}\left(1-\frac{\eps}{6}\right)\left(1+\frac{\eps}{8}\right)^{j}$ and $v_{min}\left(1+\frac{\eps}{6}\right)$ $\left(1+\frac{\eps}{8}\right)^{j}$ (Since $\frac{\eps}{8}< \frac{\eps}{6}$, the total weight of $\mathcal{W}(i,j)$ will be less than or equal to the total weight of any subset  of $V_i$ having the total value in between $\left(1+\frac{\eps}{8}\right)^j v_{min}$ and $\left(1+\frac{\eps}{8}\right)^{j+1}v_{min}$).
\end{enumerate}

The table $\mathcal{W}$ is created in Step \ref{W_max} of the algorithm. This step uses Theorem \ref{the:sub_value_max_knapsack}. We get both above properties because of the guarantee of Theorem \ref{the:sub_value_max_knapsack}.

Let $\mathcal{T}$ be a perfect binary tree with $\NoOfCategories$ leaf nodes. For simplicity, assume that $\NoOfCategories$ is power of $2$. Although, we can prove the same result by slight modification of the proof when $\NoOfCategories$ is not power of $2$. The total number of nodes in $\mathcal{T}$ will be $2\NoOfCategories-1$.  Each node in $\mathcal{T}$ could be represented by an index number from $1$ to $2\NoOfCategories-1$ with root at index $1$.
The node at an index $i$ has an left child at an index $2i$ and right child at an index $2i+1$, $\forall i\in[\NoOfCategories-1]$. Let the leaf node at an index $(\NoOfCategories-1)+i$ represent the category $i$, $\forall i\in [\NoOfCategories]$. Let $\mathcal{T}(i)$ denote the set of categories represented by the node $i$ of $\mathcal{T}$. Specifically, $\mathcal{T}(\NoOfCategories-1+i)=\{i\}$, $\forall i\in[m]$. Also, $\mathcal{T}(i)=\mathcal{T}(2i)\cup\mathcal{T}(2i+1)$, $\forall i\in [\NoOfCategories-1]$.
\paragraph{Properties of $\mathcal{X}$.}We claim that the entry $\mathcal{X}(i,j)$,  $\forall i\in[2\NoOfCategories-1]$, $\forall j \in X_{range}$, indicates the weight of a subset of $\cup_{k\in \mathcal{T}(i)} V_{k}$ that satisfies three properties mentioned below. The entry of $\mathcal{X}$ also corresponds to the  respective subset. The entry of $\mathcal{X}$ could be $\infty$, which indicates no subset. We use the notation $\mathcal{X}(i,j)$ to indicate both the entry and the subset. 
\begin{enumerate}
	\item \label{X_lower_value_min_first} If the entry $\mathcal{X}(i,j)$ is finite, then
	\[\sum_{k\in \mathcal{T}(i)} \sum_{r \in \mathcal{X}(i,j)\cap V_{k}} v^{(k)}_{r} \geq \left(1-\frac{\eps}{2}\right)\left(1+\eps{'}\right)^{j}v_{min}\].
	\item \label{X_lower_value_min_second}
	If the entry $\mathcal{X}(i,j)$ is finite, then $\forall k \in \mathcal{T}(i)$, \[\left(1-\frac{\eps}{2}\right)\left(1+\frac{\eps}{8}\right)^{-1}l_{k}^{v}\leq  \sum_{r \in  \mathcal{X}(i,j)\cap V_k} v_{r}^{(k)}  \leq \left(1+\frac{\eps}{2}\right)\left(1+\frac{\eps}{8}\right)u_{k}^{v}.\]
	
	\item \label{X_lower_value_min_third} Let $i$ be a node of $\mathcal{T}$, $\forall i \in [2\NoOfCategories-1]$, having the distance $t$ from leaves, $ t\in\left[\log_{2}\NoOfCategories\right]\cup \{0\}$.  For all $O{'} \subseteq \cup_{k\in \mathcal{T}(i)} V_{k}$ having the total value at least $\left(1+\eps{'}\right)^{j+t}v_{min}$, and $ l_{k}^{v} \leq \sum_{r \in  O{'}\cap V_{k}}v_{r}^{(k)}\leq u_{k}^{v}$, $\forall k \in \mathcal{T}(i)$, the total weight of the subset  $\mathcal{X}(i,j)$ is at most the total weight of $O{'}$. 
\end{enumerate}
The total value of any optimal subset is at least $L$. Because of the condition $v_{min}\left(1+\eps\right)^{j+\log_2 \NoOfCategories} \geq L$ in Step \ref{Output_value_min}, we can apply Property \ref{X_lower_value_min_third} of $\mathcal{X}$ to any optimal subset and $S$. This proves that the the total weight of $S$ is at most the total weight of an optimal solution. By Property \ref{X_lower_value_min_first} of $\mathcal{X}$ and the condition $v_{min}\left(1+\eps\right)^{j+\log_2 \NoOfCategories} \geq L$ in Step \ref{Output_value_min}, the total value of $S$ is at least
$\left(1-\frac{\eps}{2} \right)\left(1+\eps{'}\right)^{-\log_2 \NoOfCategories}L$. By Step \ref{eps'_value_min} of the algorithm, this number is at least 
\[\left(1-\frac{\eps}{2} \right)\left(1+\frac{3\eps}{8}\right)^{-1} L >  \left(1-\frac{\eps}{2} \right)\left(1-\frac{3\eps}{8}\right) L > (1-\eps)L. \] \par
If the entry $\mathcal{X}(1,j)$ is finite, Property \ref{X_lower_value_min_second} of $\mathcal{X}$ implies that the total value of items from $V_i \cap \mathcal{X}(1,j)$ is in between $\left(1-\frac{\eps}{2} \right)\left(1+\frac{\eps}{8} \right)^{-1}l_{i}^{v}$ and $\left(1+\frac{\eps}{2} \right)\left(1+\frac{\eps}{8} \right)u_{i}^{v}$ for all $i\in [\NoOfCategories]$. The total value of items from $V_i \cap \mathcal{X}(1,j)$ is at least $\left(1-\frac{\eps}{2} \right)\left(1+\frac{\eps}{8} \right)^{-1}l_{i}^{v}> \left(1-\frac{\eps}{2} \right)\left(1-\frac{\eps}{8} \right) l_{i}^{v}> (1-\eps)l_{i}^{v}$. The total value of items from $V_i \cap \mathcal{X}(1,j)$ is at most $\left(1+\frac{\eps}{2} \right)\left(1+\frac{\eps}{8} \right)u_{i}^{v} < \left(1+\frac{6}{8}\eps \right)u_{i}^{v}<\left(1+\eps \right)u_{i}^{v}$.
The algorithm is similar to the algorithm (Algorithm \ref{alg:fairness_based_value}) in \Cref{the:alg_value_max_knapsack}. The only difference is the output step (Step \ref{Output_value_min}). The running time analysis of Algorithm \ref{alg:fairness_based_value} is also applicable to this algorithm.
\end{proof}

\subsection{Fairness based on bound on weight}
\begin{problem}[$BW^{min}$]\label{prob:weight_min}
Given a set of items, each belonging to one of $\NoOfCategories$ categories, and  a lower bound $l_{i}^{w}$ and an upper bound $u_{i}^{w}$ such that $0\leq l_{i}^{w} \leq u_{i}^{w}$ for each category $i\in[\NoOfCategories]$, the goal is to find a subset that minimizes the total weight, such that the total weight of items from each category $i$ is in between the given bounds $l_{i}^{w}$ and $u_{i}^{w}$, $\forall i\in[m]$, and the total value of the subset is at least the given bound $L$. 
\end{problem}

We prove that it is $\NP$-hard to obtain the feasible solution of an instance of $BW^{min}$ (\Cref{prob:weight_min}). 

\begin{theorem}\label{hardness_result_bound_weight_min}
	There is no polynomial time algorithm which outputs a feasible solution of $BW^{min}$ (\Cref{prob:weight_min}), assuming $\PP \neq \NP$.
\end{theorem}
\begin{proof}
	Given an instance of subset sum (\Cref{subsetSum2}), we can construct an instance of $BW^{min}$ (\Cref{prob:weight_min}) in following way. Let $\NoOfCategories=1$, $l_{1}^{w}=\frac{1}{2}$, $u_{1}^{w}=1$ and $L=\frac{1}{2}$. The set $V_1$ contains items which correspond to the numbers in $I$. An item corresponding to some $a\in I$ has a value and a weight equal to $a$. This proves that if we have an algorithm that outputs a feasible solution of $BW^{min}$ (\Cref{prob:weight_min}), then we can solve subset sum (\Cref{subsetSum2}). But assuming $\PP\neq \NP$, this is not possible.	
\end{proof} 

Theorem \ref{hardness_result_bound_value_min} implies that there does not exists polynomial time algorithm for $BW^{min}$ (\Cref{prob:weight_min}). We give here an algorithm for $BW^{min}$ (\Cref{prob:weight_min}) that might violates fairness constraints for a category by a small amount.

\begin{theorem}\label{the:weight_min}
For any $\eps>0$, there exists an algorithm for $BW^{min}$ (\Cref{prob:weight_min}) that outputs a solution $S$ whose total weight is at most $1+\eps$ times of the total weight of an optimal solution, and $ (1-\eps) l_{i}^{w} \leq \sum_{j \in V_i \cap S}w^{(i)}_{j} \leq (1+\eps)u_{i}^{w} $, $\forall i \in [\NoOfCategories]$, and the total value of $S$ is at least $L$. The running time of the algorithm is  $O \left( \frac{n^2\NoOfCategories \log^3 \NoOfCategories \log_{(1+\eps)}^{3} \left(\frac{W}{w_{min}}\right)}{\eps}\right)$. 
Here $w_{min}:= \min\{w_{j}^{(i)} \mid i \in [\NoOfCategories]\; \& \; j \in  V_i \; \& \; w_{j}^{(i)}>0  \}$ and $W=\sum_{i=1}^{\NoOfCategories} \sum_{j \in V_j}w^{(i)}_{j}$. 
\end{theorem}
\begin{proof}
The algorithm for theorem is described in Algorithm \ref{alg:min_fairness_based_weight}.
The algorithm creates bundles of items from $V_i$, $\forall i \in [\NoOfCategories]$, such that the total weight of each bundle is in different ranges using Theorem \ref{the:sub_weight_max_knapsack} (Step \ref{Y_lower_max}). It stores these bundles in the table $\mathcal{Y}$. After that the algorithm combines bundles from all categories to obtain the final solution using the dynamic programming table $\mathcal{Z}$ (Step \ref{Z_max}). The total weight of each bundle is represented by some power of $\left(1+\eps{'}\right)$ in the tables $\mathcal{Y}$ and $\mathcal{Z}$. The algorithm might over calculate at most $\left(1+\eps{'} \right)$ fraction of total weight in table $\mathcal{Y}$ (Step \ref{Y_lower_max}). The total fraction of weight over calculated after combining bundles from all categories in Step \ref{Z_i_max} is at most $\left(1+\eps{'} \right)^{O(\log_{2}\NoOfCategories)}$. This is at most $\left(1+\eps\right)$ because of the choice of $\eps{'}$ in Step \ref{eps'_weight_max}. We describe the formal proof of the algorithm below.\par

The tables $\mathcal{Y}$ (Step \ref{Y_lower_min}) and $\mathcal{Z}$ (Step \ref{Z_min}) are same as the tables created by the algorithm (Algorithm \ref{alg:fairness_based_lower_weight}) in \Cref{the:weight_max_knapsack}. As proved in \Cref{the:weight_max_knapsack}, the entries in these tables satisfy following properties. 

\RestyleAlgo{boxruled}
\begin{algorithm}[!t]
	\caption{Algorithm for the $BW^{min}$ (\Cref{prob:weight_min})} \label{alg:min_fairness_based_weight}
	\textbf{Input:} The sets $V_{i}$ of items, $\forall i\in [\NoOfCategories]$. $0\leq l_{i}^{w}\leq u_{i}^{w}$ for $i\in[\NoOfCategories]$. The value lower bound $L$ and $\eps>0$. \\
	\textbf{Output:} The subset $S$ having the total weight at most $1+\eps$ times the weight of the optimal solution of $BW^{min}$ (\Cref{prob:weight_min}), such that $(1-\eps)l_{i}^{w} \leq \sum_{r \in S \cap V_i} \leq (1+\eps)u_{i}^{w}  $, $\forall i \in [\NoOfCategories]$, and the total value of $S$ is at least $L$.\\
	\begin{enumerate}
	\item \label{eps'_weight_min}	Let $\eps{'}=\left(1+\frac{3\eps}{8}\right)^{\frac{1}{\log_{2}\NoOfCategories+1}}-1$. 
			\item \label{Y_lower_min} Let $\mathcal{Y}(i,j), \ \forall i \in [\NoOfCategories]$, $\forall j \in \left[ \left \lceil \log_{1+\frac{\eps}{8}} \left(\frac{W}{w_{min}}\right) \right \rceil \right]\cup \{0\}$ be the table, where the entry $\mathcal{Y}(i,j)$ indicates the value of a subset of $V_i$ that is obtained by \Cref{the:sub_weight_max_knapsack} by  setting $V_i$ as $V{'}$, $w_{min}\left(1+\frac{\eps}{8}\right)^j$ as $w$ and $\frac{\eps}{6}$ as $\eps$ in \Cref{the:sub_weight_max_knapsack}.
			
			\item \label{Z_min}Let $\mathcal{Z}(i,j)$, $\forall i\in[2\NoOfCategories-1]$ and $\forall j \in \left[ \left \lceil \log_{1+\eps{'}} \left(\frac{W}{w_{min}}\right) \right \rceil+\log_2 \NoOfCategories \right]\cup \{0\}$, be the DP table \\constructed as follows.
			
			\begin{enumerate}
				
				\item \label{Z_1_min} $\forall i\in[\NoOfCategories]$				\begin{equation*}
				\begin{split}
				\mathcal{Z}(m-1+i,j):=\max&\left\{\mathcal{Y}(i,j{''}) \mid j{''} \in \left[ \left \lceil \log_{1+\frac{\eps}{8}} \left(\frac{W}{w_{min}}\right) \right \rceil \right]\cup \{0\} \; \& \right.\\
				&\left.  \left(1+\eps{'}\right)^{j{''}+1}w_{min} \geq l_{i}^{w}\; \& \; \left(1+\eps{'}\right)^{j{''}-1}w_{min} \leq u_{i}^{w} \right. \\ 
				& \left. \left(1+\frac{\eps}{8} \right)^{j{''}} \leq \left(1+\eps{'}\right)^{j}\right\} .
				\end{split}
				\end{equation*}
				If the set satisfying above condition is empty, then set $\mathcal{Z}(1,j)$ to $-\infty$.

				\item \label{Z_i_min}  $\forall i\in[\NoOfCategories-1]$ and any $j{'},j{''} \in \left[ \left \lceil \log_{1+\eps{'}} \left(\frac{W}{w_{min}}\right) \right \rceil +\log_2 \NoOfCategories \right]\cup \{0\}$, we have

				\begin{equation*}
				\begin{split}
				\mathcal{Z}(i,j):= \max & \left\{\mathcal{Z}(2i,j{'})+\mathcal{Y}(2i+1,j{''}) \mid\;   \left(1+\eps{'}\right)^{j} \geq \left(1+\eps{'}\right)^{j{'}}+\left(1+\eps{'}\right)^{j{''}} \right \} .
				\end{split}
				\end{equation*}	
				If the set satisfying above condition is empty, then set $\mathcal{Z}(i,j)$ to $-\infty$.
			\end{enumerate}	
				
		\item \label{bound_weight_selection_min}
		Output the subset $S$ as follows,
		\[ \argmin_{j} \left\{ \mathcal{Z}(1,j) \mid j\in  \left[ \left \lceil \log_{1+\eps{'}} \left(\frac{W}{w_{min}}\right) \right \rceil+\log_2 \NoOfCategories \right]\cup \{0\} \; \& \; \mathcal{Z}(1,j)\geq L \right\} . \]
	\end{enumerate}
	
\end{algorithm}

\paragraph{Properties of $\mathcal{Y}$. }We claim that the entry  $\mathcal{Y}(i,j)$,  $\forall i\in[\NoOfCategories]$, $\forall j \in \left[ \left \lceil \log_{1+\frac{\eps}{8}} \left(\frac{W}{w_{min}}\right) \right \rceil \right]\cup \{0\}$, indicates the value of a subset of $V_i$ that satisfies the two properties listed below. The entry of $\mathcal{Y}$ also corresponds to respective subset. The entry of $\mathcal{Y}$ could be $-\infty$, which indicates no subset. We use the notation $\mathcal{Y}(i,j)$ to indicate both the entry and the subset.

\begin{enumerate}
	\item \label{Y_lower_weight_min_1} If the entry $\mathcal{Y}(i,j)$ is finite, then the total weight of the corresponding subset is in between $\left(1-\frac{\eps}{2}\right)$ $\left(1+\frac{\eps}{8}\right)^j w_{min}$ and $\left(1+\frac{\eps}{2}\right)\left(1+\frac{\eps}{8}\right)^{j}w_{min}$.
	\item \label{Y_lower_weight_min_2} The total value of the subset corresponding to $\mathcal{Y}(i,j)$ is at least the total value of any subset of $V_i$ having the total weight in between $\left(1-\frac{\eps}{6}\right)\left(1+\frac{\eps}{8}\right)^j w_{min}$ and $\left(1+\frac{\eps}{6}\right)\left(1+\frac{\eps}{8}\right)^{j}w_{min}$. 
\end{enumerate}

The table $\mathcal{Y}$ is created in step \ref{Y_lower_max} of the Algorithm \ref{alg:fairness_based_lower_weight}. This step uses \Cref{the:sub_weight_max_knapsack} for creation of $\mathcal{Y}$. We get both the properties of $\mathcal{Y}$ because of the guarantee of \Cref{the:sub_weight_max_knapsack}.

Let $\mathcal{T}$ be a perfect binary tree with $\NoOfCategories$ leaf nodes. For simplicity, assume that $\NoOfCategories$ is power of $2$. Although, we can prove the same result by slight modification of the proof when $\NoOfCategories$ is not power of $2$. The total number of nodes in $\mathcal{T}$ will be $2\NoOfCategories-1$.  Each node in $\mathcal{T}$ could be represented by an index number from $1$ to $2\NoOfCategories-1$ with root at index $1$.
The node at an index $i$ has an left child at an index $2i$ and right child at an index $2i+1$, $\forall i\in[\NoOfCategories-1]$. Let the leaf node at an index $(\NoOfCategories-1)+i$ represent the category $i$, $\forall i\in [\NoOfCategories]$. Let $\mathcal{T}(i)$ denote the set of categories represented by the node $i$ of $\mathcal{T}$. Specifically, $\mathcal{T}(\NoOfCategories-1+i)=\{i\}$, $\forall i\in[m]$. Also, $\mathcal{T}(i)=\mathcal{T}(2i)\cup\mathcal{T}(2i+1)$, $\forall i\in [\NoOfCategories-1]$.

\paragraph{Properties of $\mathcal{Z}$. }We claim that the entry $\mathcal{Z}(i,j)$,  $\forall i\in[2\NoOfCategories-1]$, $j \in \left[ \left \lceil \log_{1+\eps{'}} \left(\frac{W}{w_{min}}\right) \right \rceil +\log_2 \NoOfCategories \right]\cup \{0\}$, indicates the value of the subset of $\cup_{k\in \mathcal{T}(i)} V_{k}$ that satisfies the following three properties. The entry of $\mathcal{Z}$ also corresponds to the  respective subset. The entry of $\mathcal{Z}$ could be $-\infty$, which indicates no subset. We use the notation $\mathcal{Z}(i,j)$ to indicate both the entry and the subset.
\begin{enumerate}
	\item \label{Z_lower_weight_min_first} If the entry $\mathcal{Z}(i,j)$ is finite, then $\sum_{k\in \mathcal{T}(i)} \sum_{r \in \mathcal{Z}(i,j)\cap V_{k}} w^{(k)}_{r} \leq \left(1+\frac{\eps}{2}\right)\left(1+\eps{'}\right)^{j}w_{min}$.
	\item \label{Z_lower_weight_min_second} If the entry $\mathcal{Z}(i,j)$ is finite, then $\left(1-\frac{\eps}{2}\right)\left(1+\frac{\eps}{8}\right)^{-1}l_{k}^{w} \leq \sum_{r \in V_{k} \cap \mathcal{Z}(i,j)} w_{r}^{(k)} \leq \left(1+\frac{\eps}{2}\right)\left(1+\frac{\eps}{8}\right)u_{k}^{w}$, $\forall k \in \mathcal{T}(i)$.
	\item \label{Z_lower_weight_min_third} Let $i$ be a node of $\mathcal{T}$, $\forall i \in [2\NoOfCategories-1]$, having the distance $t$ from leaves, $ t\in\left[\log_{2}\NoOfCategories\right]\cup \{0\}$.  For all $O{'} \subseteq \cup_{k\in \mathcal{T}(i)} V_{k}$ having the total weight at most $\left(1+\eps{'}\right)^{j-t}w_{min}$, and $ l_{k}^{w} \leq \sum_{r \in  O{'}\cap V_{k}}w_{r}^{(k)}\leq u_{k}^{w}$, $\forall k \in \mathcal{T}(i)$, the total value of the subset  $\mathcal{Z}(i,j)$ is at least the total value of $O{'}$.
\end{enumerate}

Let $\OPT$ be the weight of an optimal solution and $j\in \left[ \left \lceil \log_{1+\eps{'}} \left(\frac{W}{w_{min}}\right) \right \rceil +\log_2 \NoOfCategories \right]\cup \{0\}$ be the number such that,

\begin{equation}\label{eq:optimal_weight_j_weight_min}
    w_{min}\left(1+\eps{'}\right)^{j-\log_2\NoOfCategories-1}\leq \OPT \leq w_{min}\left(1+\eps{'}\right)^{j-\log_2\NoOfCategories}. 
\end{equation}

So we can apply Property \ref{Z_lower_weight_min_third} of $\mathcal{Z}$ to an optimal set and $\mathcal{Z}(1,j)$. So, the total value of $\mathcal{Z}(1,j)$ is at least the total value of an optimal set, which is more than $L$. So, the subset $\mathcal{Z}(1,j)$ is feasible in Step \ref{bound_weight_selection_min} of the algorithm. If $S$ in Step \ref{bound_weight_selection_min} corresponds to some $j{'}\in \left[ \left \lceil \log_{1+\eps{'}} \left(\frac{W}{w_{min}}\right) \right \rceil +\log_2 \NoOfCategories \right]\cup \{0\}$, then $j{'} \leq j$. The total weight of $S$ at most 

	\begin{align*}
&\left(1+\frac{\eps}{2}\right)\left(1+\eps{'}\right)^{j{'}}w_{min} &\textrm{(Property \ref{Z_lower_weight_max_first} of $\mathcal{Z}$)}\\
&\leq \left(1+\frac{\eps}{2}\right)\left(1+\eps{'}\right)^{j}w_{min} &\textrm{($j{'} \leq j$)}\\
&\leq \left(1+\frac{\eps}{2}\eps\right)\left(1+\eps{'}\right)^{\log_2 \NoOfCategories+1}\OPT &\textrm{(Inequality \ref{eq:optimal_weight_j_weight_min})}\\
&\leq \left(1+\frac{\eps}{2}\right)\left(1+\frac{3\eps}{8}\right)\OPT &\textrm{(Step \ref{eps'_weight_min} of the algorithm)}\\
&<\left(1+\eps\right)\OPT.
	\end{align*}

If the entry $\mathcal{Z}(1,j)$ is finite, Property \ref{X_lower_value_max_second} of $\mathcal{Z}$ implies that the total weight of items from $V_i \cap \mathcal{Z}(1,j)$ is in between $\left(1-\frac{\eps}{2} \right)\left(1+\frac{\eps}{8} \right)^{-1}l_{i}^{w}$ and $\left(1+\frac{\eps}{2} \right)\left(1+\frac{\eps}{8} \right)u_{i}^{w}$ for all $i\in [\NoOfCategories]$. The total weight of items from $V_i \cap \mathcal{Z}(1,j)$ is at least $\left(1-\frac{\eps}{2} \right)\left(1+\frac{\eps}{8} \right)^{-1}l_{i}^{w}> \left(1-\frac{\eps}{2} \right)\left(1-\frac{\eps}{8} \right) l_{i}^{w}> (1-\eps)l_{i}^{w}$. The total weight of items from $V_i \cap \mathcal{Z}(1,j)$ is at most $\left(1+\frac{\eps}{2} \right)\left(1+\frac{\eps}{8} \right)u_{i}^{w} < \left(1+\frac{6}{8}\eps \right)u_{i}^{w}<\left(1+\eps \right)u_{i}^{w}$. \par

The algorithm is similar to the algorithm (Algorithm \ref{alg:fairness_based_lower_weight}) in \Cref{the:weight_max_knapsack}. The only differences are the output step (Step \ref{bound_weight_selection_min}) and the size of tables $\mathcal{Y}$ and $\mathcal{Z}$. So, the running time analysis of algorithm is similar to the running time of Algorithm \ref{alg:fairness_based_lower_weight}.

\end{proof}

%% file: open_problems.tex
\section{Conclusion}
In this paper, we studied various fairness notions for knapsack problems. Studying fairness notions in related problems such as multiple knapsack problem \cite{Jansen}, multidimensional knapsack problem \cite{Frieze}, submodular knapsack problem \cite{Lee}, is an interesting open problem. \par

\paragraph{Acknowledgements.} AL was supported in part by SERB Award ECR/2017/003296 and a Pratiksha Trust Young Investigator Award.